\documentclass[11pt, letterpaper]{article}

\usepackage{graphicx}
\usepackage{epstopdf}
\usepackage{amsmath,amssymb,amsfonts}
\usepackage{setspace}
\usepackage{epstopdf}
\usepackage{color}
\usepackage[letterpaper,left=1in,right=1in,top=1in,bottom=1in]{geometry}
\usepackage{multirow}
\usepackage{longtable}
\usepackage{rotating}
\usepackage{mathrsfs}
\usepackage[ruled,vlined]{algorithm2e}
\usepackage{algpseudocode}
\usepackage{booktabs}
\usepackage{ntheorem}
\usepackage{arydshln}
\newtheorem{rmk}{Remark}
\newtheorem{assumption}{Assumption}

\newtheorem{lemma}{Lemma}
\newtheorem{proposition}{Proposition}

\newtheorem{theorem}{Theorem}

\newtheorem{definition}{Definition}
\newtheorem*{proof}{Proof}
\usepackage[definethebibliography]{easybib}
\title{\Large Partition-based distributed extended Kalman filter for large-scale nonlinear processes with application to chemical and wastewater treatment processes}

\author{
\centerline{\normalsize Xiaojie Li$^{a}$, Adrian Wing-Keung Law$^{b,c}$, Xunyuan Yin$^{a,}$\thanks{Corresponding author: X. Yin. Tel: (+65) 6316 8746. Email: xunyuan.yin@ntu.edu.sg.}
}
\vspace{5mm}\\
\centerline{\small $^{a}$School of Chemistry, Chemical Engineering and Biotechnology, Nanyang Technological University,}\\
\centerline{\small 62 Nanyang Drive, 637459, Singapore}\\
\centerline{\small $^{b}$School of Civil and Environmental Engineering, Nanyang Technological University,}\\
\centerline{\small 50 Nanyang Avenue, 639798, Singapore}\\
\centerline{\small $^{c}$Environmental Process Modelling Centre, Nanyang Environment and Water Research Institute (NEWRI), }\\
\centerline{\small Nanyang Technological University,}\\
\centerline{\small 1 CleanTech Loop, CleanTech One, \#05-01, 637141, Singapore}
}
\allowdisplaybreaks
\begin{document}

\date{}
\maketitle
\doublespacing
%
\begin{abstract}
In this paper, we address a partition-based distributed state estimation problem for large-scale general nonlinear processes by proposing a Kalman-based approach. First, we formulate a linear full-information estimation design within a distributed framework as the basis for developing our approach. Second, the analytical solution to the local optimization problems associated with the formulated distributed full-information design is established, in the form of a recursive distributed Kalman filter algorithm. Then, the linear distributed Kalman filter is extended to the nonlinear context by incorporating successive linearization of nonlinear subsystem models, and the proposed distributed extended Kalman filter approach is formulated. We conduct rigorous analysis and prove the stability of the estimation error dynamics provided by the proposed method for general nonlinear processes consisting of interconnected subsystems. A chemical process example is used to illustrate the effectiveness of the proposed method and to justify the validity of the theoretical findings. In addition, the proposed method is applied to a wastewater treatment process for estimating the full state of the process with 145 state variables.
\end{abstract}
\noindent{\bf Keywords:}Distributed state estimation, partition-based framework, extended Kalman filter, nonlinear process
\section*{Introduction}
A large-scale industrial process typically exhibits significant nonlinearity, and consists of multiple physical components that are interconnected tightly with each other through material, energy, and information flows. Advanced control systems play a crucial role in monitoring and governing the operation of industrial processes for ensuring the operating safety, meeting the strict environmental regulations, and pursuing desired production objectives. However, due to the high nonlinearity, large scales, as well as the strong intrinsic interactions among different physical components, the development and commissioning of advanced control systems for large-scale complex processes have been a challenging task, and conventional centralized or decentralized paradigms may not suffice to manage these complexities \cite{christofides2013distributed, daoutidis2016sustainability,stewart2010cooperative}.

As an emerging framework, the distributed control architecture, in particular, the partition-based distributed framework holds the promise to handle plant-wide control problems for large-scale industrial processes, while offering a high level of fault tolerance, scalability, computational efficiency, and maintenance flexibility \cite{christofides2013distributed,daoutidis2016sustainability,yin2018forming}.
A distributed solution decomposes large processes into smaller subsystems. Accordingly, the complex plant-wide control-related problem is divided into smaller scalable sub-problems associated with the configured subsystems. Then, multiple decision-making agents that communicate with each other in real-time can be developed to collaboratively handle sub-problems in parallel. A complete distributed process control system shall involve distributed model predictive control and distributed state estimation, the latter of which provides real-time full-state information, which is essential for distributed control to execute algorithms to generate more-informed control decisions. Distributed model predictive control (DMPC) for large-scale processes has been extensively studied \cite{daoutidis2019decomposition,tang2018optimal,chen2020machine,tang2019distributed}. This work is focused on the partition-based distributed state estimation counterpart.

Distributed moving horizon estimation (DMHE) has been considered a promising solution to distributed state estimation problems, due to its ability to simultaneously address the nonlinearity of a process and various constraints in state and disturbance variables \cite{yin2018forming,farina2010distributed,farina2012distributed,zhang2013distributed,schneider2015convergence, yin2022event,pourkargar2019distributed}.
In Farina et al. \cite{farina2010distributed}, a consensus-based DMHE with constraints on state variables and noise was presented; this method was later extended to nonlinear constrained systems in Farina et al. \cite{farina2012distributed}. In another line of work, partition-based DMHE where the local estimators are developed based on decomposed subsystem models has been studied. In Schneider and Marquardt \cite{schneider2015convergence} and Yin and Huang \cite{yin2022event}, iterative DMHE algorithms were proposed for linear systems, with local estimators developed based on the models of the interacting subsystems. Partition-based DMHE methods were proposed for nonlinear process networks in Yin et al \cite{yin2018forming} and Zhang and Liu \cite{zhang2013distributed}. While DMHE offers the advantage of explicitly handling hard constraints, it does have limitations when applied to highly nonlinear processes. This is because even though the global estimation problem is partitioned into smaller sub-problems, each sub-problem typically remains a non-convex optimization problem. From an application standpoint, there is no guarantee that the subsystem states can be optimally reconstructed per the theoretical analysis, which may lead to unsatisfactory estimates of the process states. Furthermore, with poor initial guesses that deviate significantly from the initial conditions of the subsystems, it can be more likely that diverging estimation results will be produced by DMHE for nonlinear systems. To address these limitations associated with DMHE, a partition-based distributed extended Kalman filter design may be considered a viable alternative. Kalman-based recursive distributed estimation solutions can be more attractive when constraints are not significant and the computational efficiency of distributed estimation needs to be increased to ensure the timely availability of the estimates to the controllers. Additionally, it enables the solving of the distributed nonlinear MPC, which is crucial for facilitating the commissioning and implementation of distributed control systems for complex industrial processes.

Distributed Kalman Filter (DKF) can be classified into two categories according to the architecture: consensus-based DKF and partition-based DKF.
In consensus-based DKF, multiple estimators are designed based on a global model, and each estimator provides estimates of the states of the entire system based on local sensors. These estimates provided by the local estimators are then combined via a consensus protocol that provides a way for each subsystem to reach agreement \cite{olfati2007distributed,kim2016distributed,olfati2005distributed,wu2018distributed, marelli2021distributed}.
{\color{black}Consensus-based DKF has shown great potential in a variety of applications, including sensor networks \cite{olfati2007distributed, wu2018distributed} and multi-robot systems \cite{marin2019event}.}
As is completely different from consensus-based solutions, partition-based DKF divides the global model into smaller subsystem models, and designs local estimators based on corresponding subsystem models. Each local estimator provides estimates of the corresponding subsystem, which is a subset of the state variables. This method is particularly useful when the system scale is excessively large, or when the entire system consists of spatially distributed and physically interconnected operating units. In the literature, several attempts have been made to develop partition-based DKF methods for large-scale systems.
In Sun et al \cite{sun2016dynamic}, an initial attempt was made on partition-based DKF, whereby subsystem coupling was assumed to exist only in the output measurement equations, yet the future state of each subsystem was assumed to be dependent only on the current state of the local subsystem.
In Marelli et al \cite{marelli2018distributed}, subsystem interactions were incorporated into the state equations of each subsystem, and a DKF algorithm that can asymptotically converge to its centralized counterpart was proposed. However, the coupling between subsystems was assumed to only arise from the sensor measurements of the neighboring subsystems, while unmeasured states were excluded from the subsystem interactions; this has significantly limited the application of this method to more general interconnected systems.
In Zhang et al \cite{zhang2021distributed}, a DKF approach was proposed for an interconnected dynamical system with state interactions in the subsystem models. In this work, decoupling Kalman gains were used to minimize the impact of the estimation error caused by subsystem interactions. Meanwhile, it is worth mentioning that Zhang et al \cite{zhang2021distributed} relies on a directed acyclic graph assumption that requires the output measurement matrix for each subsystem is full column rank, which may be difficult to satisfy in process control applications. Consequently, this approach may not be suitable for being extended to nonlinear distributed state estimation problems.
Despite the prevalence of Kalman-based algorithms, to our knowledge, results on partition-based distributed extended Kalman filter algorithms for general nonlinear systems have been very limited. In Rashedi et al \cite{rashedi2017distributed}, a distributed adaptive high gain extended Kalman filter algorithm was proposed.
This design requires an invertible coordinate transformation of the original state of each subsystem to a canonical form. However, as the number of states of each subsystem grows, establishing this transformation becomes increasingly difficult, which has significantly limited the applicability of this algorithm to large-scale nonlinear processes.
A distributed EKF was developed for the estimation of a wastewater treatment process with 78 quality state variables in Zeng et al \cite{zeng2016distributed}, whereby each local estimator only leverages the sensor measurements of the local subsystem, and the stability of the estimation error dynamics remains unsolved.
Based on above consideration, we aim to address the limitations of the existing DEKF algorithms in this work.

In this paper, we propose a scalable and efficient partition-based distributed extended Kalman filter method for general nonlinear processes consisting of subsystems interconnected through state interactions. The objective is achieved by conducting the following key steps: 1) we formulate a distributed full-information estimation design based on the subsystem models of a general linear system; 2) inspired by the equivalence between unconstrained centralized full-information estimation and centralized Kalman filter, we obtain analytical solutions to the local estimators of the full-information estimation, and the analytical solutions are in the form of a distributed Kalman filter algorithm; 3) the nonlinear extension of the proposed distributed Kalman filter is conducted to form the proposed partition-based DEKF approach for general nonlinear processes. The stability of the proposed method is proven. The effectiveness and the advantages of the proposed method are illustrated via a numerical example, as well as two application examples related to chemical and wastewater treatment processes.

\section*{Preliminaries}
\subsection*{Notation}

${\rm{col}}\{x_1, x_2, \ldots, x_n\}$ represents a column vector consisting of $x_1, x_2, \ldots, x_n.$ 
$\text{diag}\left\{A_1,A_2,\ldots,A_n\right\}$ denotes a block diagonal matrix where the blocks on the main diagonal are matrices $A_i$, $i=1,\ldots, n$. $[A_{il}]$ represents a block matrix where $A_{il}$ is the submatrix in the $i$th row and the $l$th column.
{\color{black}$\lambda_{max}(A)$ denotes the maximal eigenvalue of matrix $A$. $\|A\|$ represents the Euclidean norm of real matrix $A$, which is computed as $\|A\|=\sqrt{\lambda_{max}(A^{\mathrm{T}}A)}$.}
$\|x\|^2$ and $\|x\|^2_{P}$ denote the Euclidean norm and the weighted Euclidean norm of a vector, which are computed as $\|x\|^{2}=x^{\mathrm{T}}x$ and $\|x\|^{2}_{P}=x^{\mathrm{T}}Px$, respectively.
$\varepsilon(x,y)$ denotes the covariance between two variables $x$ and $y$. $x_{j|k}$ denotes the value of variable $x$ for time instant $j$ obtained at time instant $k$.
Let $\mathbb{R}_{\geq0}:=[0,\infty)$;
A function $\gamma:\mathbb{R}_{\geq0}\rightarrow\mathbb{R}_{\geq0}$ is said to belong to class $\mathcal{K}$ if it is continuous, strictly increasing and $\gamma(0)=0$.
In addition, function $\gamma:\mathbb{R}_{\geq0}\rightarrow\mathbb{R}_{\geq0}$ belongs to $\mathcal{K}_{\infty}$, if it is a class $\mathcal{K}$ function and is radially unbounded.
A function $\beta:\mathbb{R}_{\geq0}\times\mathbb{R}_{\geq0}\rightarrow\mathbb{R}_{\geq0}$ is said to belong to class $\mathcal{K}\mathcal{L}$, if for each fixed $t\geq0$,  $\beta(s, t)\in\mathcal{K}$, and for every fixed $s\geq0$, the function $\beta(s,t)$ is decreasing and satisfies $\beta(s,t)\rightarrow0$ as $t\rightarrow\infty$. 

\subsection*{System description}\label{sec:2.2}
We consider a class of general nonlinear systems that consist of $n$ interconnected subsystems; the dynamics of each subsystem $i$, $i\in \mathbb{N}$, can be described by the following discrete-time state-space form:
\begin{subequations}\label{model_nonlinear}
  \begin{align}
    x^{i}_{k+1} & = f_{i}(x^{i}_{k}, X^{i}_{k})+w_{k}^{i} \label{model_nonlinear1} \\
    y_{k}^{i} & =h_{i}(x_{k}^{i})+v_{k}^{i}\label{model_nonlinear2}
  \end{align}
\end{subequations}
{\color{black}where $x^{i}_{k}\in\mathbb{R}^{n_{x^{i}}}$ is the state vector of $i$th subsystem; $X^{i}_{k}\in \mathbb{R}^{n_{X^{i}}}$ is a vector of the states of all the subsystems whose states have direct influence on the dynamics of subsystem $i$; $y^{i}_{k}\in\mathbb{R}^{n_{y^i}}$ is vector of sensor measurements for subsystem $i$; $w^{i}_{k}\in\mathbb{R}^{n_{x^i}}$ denotes the vector of state disturbances for subsystem $i$; $v_{k}^{i}\in\mathbb{R}^{n_{y^{i}}}$ is sensor measurement noise for subsystem $i$; $f_{i}$ is a vector-valued nonlinear function that characterizes the dependence of the future state of subsystem $i$ on the current state of the local subsystem and the current states of the interacting subsystems; $h_{i}$ is the function of sensor measurements $y^{i}_{k}$ for subsystem $i$, $i\in \mathbb{N}$. It is assumed that $\varepsilon(w_{k}^{i}, w_{k}^{j})=0$ and $\varepsilon(v_{k}^{i}, v_{k}^{j})=0$, $\forall i \neq j$.}

By creating augmented vectors $x_{k}={\rm{col}}\{x_{k}^1, x_{k}^2,\ldots, x_{k}^{n}\}\in\mathbb{R}^{n_x}$, $y_{k}={\rm{col}}\{y_{k}^1, y_{k}^2,\ldots$, $y_{k}^{n}\}\in\mathbb{R}^{n_{y}}$, $w_{k}={\rm{col}}\{w_{k}^1, w_{k}^2,\ldots, w_{k}^{n}\}\in\mathbb{R}^{n_{x}}$, and $v_{k}={\rm{col}}\{v_{k}^1, v_{k}^2,\ldots, v_{k}^{n}\}\in\mathbb{R}^{n_{y}}$, an aggregated model of
$n$ subsystems in the form of \eqref{model_nonlinear} can be obtained as follows:
\begin{subequations}\label{model_nonlinear_centralized}
\begin{align}
    x_{k+1} & = f(x_{k})+w_{k} \label{model_nonlinear_centralized1} \\
    y_{k}& =h(x_{k})+v_{k}\label{model_nonlinear_centralized2}
\end{align}
\end{subequations}
where $f(x_{k})= \mathrm{col}\big\{f_{1}(x^{1}_{k}, X^{1}_{k}), f_{2}(x^{2}_{k}, X^{2}_{k}),\ldots, f_{n}(x^{n}_{k}, X^{n}_{k})\big\}$ and  $h(x_{k})= \mathrm{col}\big\{h_{1}(x^{1}_{k}), h_{2}(x^{2}_{k}), \ldots, h_{n}(x^{n}_{k})\big\}$.
The state-space model in \eqref{model_nonlinear_centralized} characterizes the nonlinear dynamics of the entire large-scale process.

\subsection*{Problem formulation}
The objective of this work is to propose a partition-based distributed state estimation approach based on the subsystem models in the form of \eqref{model_nonlinear}, which can be used to estimate the state of the entire process \eqref{model_nonlinear_centralized} in real-time. To achieve this objective, we aim to propose a partition-based distributed Kalman filter algorithm, and then extend it to the nonlinear context to account for the estimation of the nonlinear process in \eqref{model_nonlinear_centralized}.

Based on the above consideration, we consider the linear counterpart of \eqref{model_nonlinear_centralized} which consists of $n$ linear subsystem models. The model of the $i$th linear subsystem, $i\in\mathbb{N}$, is described as follows:
\begin{subequations}\label{model}
  \begin{align}
   x_{k+1}^{i}&=A_{ii}x_{k}^{i}+\sum_{l\in\mathbb{N}\setminus\{i\}}A_{il}x_{k}^{l}+w_{k}^{i}\label{modeli1}\\
   y_{k}^{i}&=C_{ii}x_{k}^{i}+v_{k}^{i}\label{modeli2}
\end{align}
\end{subequations}
where $A_{ii}$, $A_{il}$, and $C_{ii}$, $i\in\mathbb{N}$, $l\in\mathbb{N}\setminus\{i\}$, are subsystem matrices and of compatible dimensions.

We will design local Kalman filter-based state estimators based on the subsystem models in \eqref{model}, and integrate them via information exchange to form a distributed Kalman filter algorithm, which will be further extended to form a nonlinear version of the proposed algorithm for partition-based distributed state estimation of the considered nonlinear process in \eqref{model_nonlinear}.
Note that by following a similar procedure as adopted in the ``System description" Section, the linear subsystem models in \eqref{model} can be concatenated to form a compact global model as follows:
\begin{subequations}\label{cmodel}
  \begin{align}
   x_{k+1} & =Ax_{k}+w_{k} \label{cmodel1}\\
   y_{k} & =Cx_{k}+v_{k} \label{cmodel2}
\end{align}
\end{subequations}
where $A=[A_{il}]$ and $C=[C_{il}]$ with $C_{il}=0$ if $i\neq l$, $\forall~i, l\in \mathbb{N}$.

\subsection*{Centralized full-information estimation}
We briefly review centralized full-information estimation, which is leveraged as the basis of developing the partition-based distributed Kalman filter algorithm.

At each time step $k$, the centralized full-information estimation for global linear system \eqref{cmodel} is formulated as the following batch least-squares problem \cite{findeisen1997moving}:
\begin{subequations}\label{cfie}
\begin{align}
&\min_{\hat{x}_{0|k},\hat{w}_{0|k}, \ldots, \hat{w}_{k-1|k}} \Phi_{k} =\frac{1}{2}\Big(\big\|\hat{x}_{0|k}-\bar{x}_{0}\big\|_{P_{0}^{-1}}^{2}+\sum_{j=0}^{k-1}\|\hat{w}_{j|k}\|^2_{Q^{-1}}+\sum_{j=0}^{k}\|\hat{v}_{j|k}\|^2_{R^{-1}}\Big)\label{Cfie1}\\
&\quad\quad\quad\quad\quad\quad\quad\text{s.t.}\quad  \hat{x}_{j+1|k}=A\hat{x}_{j|k}+\hat{w}_{j|k},~j=0,1,\ldots,k-1\label{Cfie2}\\
&\quad\quad\quad\quad\quad\quad\quad\quad\quad\quad\quad\, y_{j} = C\hat{x}_{j|k}+\hat{v}_{j|k},~j=0,1,\ldots,k
\end{align}
\end{subequations}
where $\hat{x}_{j|k}$, $j=0,\ldots, k-1$, is an estimate of the state $x_{j}$ at time instant $j$ calculated at time instant $k$; $\bar{x}_{0}$ is \emph{a priori} estimate of state $x_{0}$; $\hat{w}_{j|k}$ is an estimate of $w_{j}$ obtained at time instant $k$; $\hat{v}_{j|k}$ is an estimate of measurement noise $v_{j}$ computed at time instant $k$; {\color{black}$P_{0}$, $Q$, and $R$ are weighting matrices which are typically made positive definite and block diagonal.}
\section*{Distributed full-information estimation}
In this section, the centralized full-information estimation problem is decomposed into local full-information estimation problems corresponding to the subsystems in \eqref{model}.
\subsection*{Construction of local objective functions}
First, we partition the global objective function $\Phi_{k}$ \eqref{Cfie1} into $\Phi_{k}^{i}$, $i\in\mathbb{N}$, such that $\Phi_{k}=\sum_{i=1}^{n}\Phi_{k}^{i}$. Each partitioned objective function is with the following form:
\begin{equation}\label{eq:3}
  \Phi_{k}^{i} =\frac{1}{2}\Big(\big\|\hat{x}^{i}_{0|k}-\bar{x}^{i}_{0}\big\|_{P_{i,0}^{-1}}^{2}+\sum_{j=0}^{k-1}\|\hat{w}^{i}_{j|k}\|^2_{Q_{i}^{-1}}+\sum_{j=0}^{k}\|\hat{v}^{i}_{j|k}\|^2_{R_{i}^{-1}}\Big)
\end{equation}
where $\hat{x}^{i}_{0|k}$ is an estimate of state $x_{0}^{i}$ of the subsystem $i$ calculated at time instant $k$; $\bar{x}^{i}_{0}$ is an initial guess of the state of the $i$th subsystem; $\hat{w}^{i}_{j|k}$ is an estimate of disturbances $w^{i}_{j}$ of the $i$th subsystem obtained at time instant $k$; $\hat{v}^{i}_{j|k}$ is an estimate of measurement noise $v_{j}^{i}$ calculated at time instant $k$. \eqref{eq:3} for subsystem $i$ only contains sensor measurement information from the local subsystem. Meanwhile, the sensor measurements of the interacting subsystems typically contain information valuable for reconstructing the states of the local subsystem. Following the formulations of the local objective functions for distributed moving horizon estimation in Schneider and Marquardt \cite{schneider2015convergence} and Li et al \cite{li2023iterative}, the sensor measurement information of the interacting subsystems of the $i$th subsystem is integrated with the partitioned objective function $\Phi_{k}^{i}$ in \eqref{eq:3} to form an individual objective function for the local estimator for the $i$th subsystem as follows:
 \begin{align}\label{eq:4}
 \bar{\Phi}_{k}^{i} &= \Phi_{k}^{i}+\frac{1}{2}\sum_{l\in\mathbb{N}\setminus\{i\}}\sum_{j=0}^{k}\|\hat{v}_{j|k}^{l}\|^{2}_{R_{l}^{-1}}\nonumber\\
 & = \frac{1}{2}\Big(\big\|\hat{x}^{i}_{0|k}-\bar{x}^{i}_{0}\big\|_{P_{i,0}^{-1}}^{2}+\sum_{j=0}^{k-1}\|\hat{w}^{i}_{j|k}\|^2_{Q_{i}^{-1}}+\sum_{j=0}^{k}\|\hat{v}^{i}_{j|k}\|^2_{R_{i}^{-1}}+\sum_{l\in\mathbb{N}\setminus\{i\}}\sum_{j=0}^{k}\|\hat{v}_{j|k}^{l}\|^{2}_{R_{l}^{-1}}\Big)\nonumber\\
 &=\frac{1}{2}\Big(\big\|\hat{x}_{0|k}^{i}-\bar{x}_{0}^{i}\big\|_{P^{-1}_{i,0}}^{2}+\sum_{j=0}^{k-1}\|\hat{w}^{i}_{j|k}\|^2_{Q_{i}^{-1}}+\sum_{j=0}^{k}\|\hat{\mathbf{v}}^{i}_{j|k}\|^2_{R^{-1}}\Big)
\end{align}
where $\mathbf{\hat{v}}_{j|k}^{i}=\mathrm{col}\{\hat{v}^{1}_{j|k}, \hat{v}^{2}_{j|k}, \ldots, \hat{v}^{n}_{j|k}\}$, which includes the estimates of measurement noise of the $i$th subsystem calculated at time instant $k$ and the estimates of other measurement noise from neighboring subsystems calculated by subsystem $i$ at time instant $k$; $R=\mathrm{diag}\{R_{1}, R_{2}, \ldots, R_{n}\}$.
\subsection*{Distributed full-information estimation formulation}
Based on the individual local objective functions in \eqref{eq:4}, a distributed full-information design can be formulated,
 where the estimator for the $i$th subsystem solves a local optimization problem as follows:
\begin{subequations}\label{fie}
\begin{align}
&\quad\quad\quad\quad\quad\quad\min_{\hat{x}^{i}_{0|k},\hat{w}^{i}_{0|k}, \ldots, \hat{w}^{i}_{k-1|k}} \bar{\Phi}_{k}^{i} \label{fie1}\\
&\text{s.t.}~\hat{x}_{j+1|k}^{i}=A_{ii}\hat{x}^{i}_{j|k}+\sum_{l\in\mathbb{N}\setminus\{i\}}A_{il}\hat{x}_{j|k-1}^{l}+\hat{w}^{i}_{j|k},~ j= 0,\ldots,k-1 \label{fir2}\\   
& \quad\quad\quad~\, y_{0}=C_{[:,i]}\hat{x}_{0|k}^{i}+\sum_{l\in\mathbb{N}\setminus\{i\}}C_{[:,l]}\hat{x}^{l}_{0|-1}+\hat{\mathbf{v}}_{0|k}^{i}\\
& \quad\quad\quad~\, y_{j}=
C_{[:,i]}\hat{x}^{i}_{j|k}+\sum_{l\in\mathbb{N}\setminus\{i\}}C_{[:,l]}A_{li}\hat{x}^{i}_{j-1|k}+\sum_{l\in\mathbb{N}\setminus\{i\}}\sum_{m\in\mathbb{N}\setminus\{i\}}C_{[:,l]}A_{lm}\hat{x}^{m}_{j-1|j-1}+\hat{\mathbf{v}}_{j|k}^{i},~ j= 1,\ldots,k
\end{align}
\end{subequations}
{\color{black}
The distributed full-information estimation design in \eqref{fie} involves multiple local estimators, each of which solves an optimization problem at every sampling instant as new sensor measurements are available.  Accordingly, each local estimator can provide an estimate of the state of the corresponding subsystem. It is worth mentioning that, within a full-information context, each local estimator solves an optimization problem with an estimation window based on the measurements from the initial time instant up to the current time instant. From a theoretical perspective \cite{allan2021robust, rawlings2012postface}, the above distributed full-information estimation has the potential to provide good estimates with stability guarantee for system \eqref{cmodel}. However, as the estimation window grows over time, the individual optimization problems for the subsystems will become intractable.}
In the next section, we propose a distributed Kalman filter algorithm, where each local estimator represents the analytical solution to \eqref{fie}. In this manner, the estimation performance of \eqref{fie} is preserved, and the intractability issue associated with online solving \eqref{fie} is bypassed.


\section*{Distributed Kalman filter algorithm}
In this section, we propose a partition-based distributed Kalman filter algorithm by leveraging the distributed full-information design in \eqref{fie}. A schematic diagram of the distributed Kalman filter scheme is shown in Figure~\ref{fig:dkf}. This distributed scheme has a total of $n$ local Kalman filters (estimators) to account for the state estimation of the $n$ subsystems in the form of \eqref{model}, and the local filters (estimators) exchange subsystem sensor measurements at each sampling instant to coordinate their estimates of the subsystem states. Note that the proposed distributed Kalman filter (DKF) algorithm will be extended to account for nonlinear estimation, and the distributed architecture in Figure~\ref{fig:dkf} remains the same in the nonlinear context.

First, let us consider distributed full-information estimation problem \eqref{fie} at $k=N=1$ for estimating $\hat{x}_{1|1}^{i}$. Replacing $P_{i,0}^{-1}$ with $P_{i,0|-1}^{-1}$ and replacing $\bar{x}_{0}^{i}$ with $\hat{x}_{0|-1}^{i}$, the full-information optimization problem associated with the $i$th subsystem in \eqref{fie} can be rewritten as
\begin{subequations}\label{fie00}
\begin{align}
&\min_{\hat{x}^{i}_{0|1},~ \hat{w}_{0|1}^{i}} \bar{\Phi}_{1}^{i}=\frac{1}{2}\Big(\big\|\hat{x}_{0|1}^{i}-\hat{x}_{0|-1}^{i}\big\|_{P^{-1}_{i,0|-1}}^{2}+\|\hat{w}_{0|1}^{i}\|_{Q_{i}}^{2}+\|\hat{\mathbf{v}}_{0|1}^{i}\|^2_{R^{-1}}+\|\hat{\mathbf{v}}_{1|1}^{i}\|^2_{R^{-1}}\Big)\nonumber\\
&\quad\quad\text{s.t.}\quad y_{0} = C_{[:,i]}\hat{x}^{i}_{0|1}+\sum_{l\in\mathbb{N}\setminus\{i\}}C_{[:,l]}\hat{x}_{0|-1}^{l}+\hat{\mathbf{v}}_{0|1}^{i}\label{upd:con1}\\
&\quad\quad\quad~~\hat{x}_{1|1}^{i} = A_{ii}\hat{x}_{0|1}^{i}+\sum_{l\in\mathbb{N}\setminus\{i\}}A_{il}\hat{x}_{0|0}^{l}+\hat{w}_{0|1}^{i}\label{upd:con3}\\
& \quad\quad\quad\quad\quad~y_{1} =  C_{[:,i]}\hat{x}^{i}_{1|1}+\sum_{l\in\mathbb{N}\setminus\{i\}}C_{[:,l]}A_{li}\hat{x}^{i}_{0|1}+\sum_{l\in\mathbb{N}\setminus\{i\}}\sum_{j\in\mathbb{N}\setminus\{i\}}C_{[:,l]}A_{lj}\hat{x}^{j}_{0|0}+\hat{\mathbf{v}}_{1|1}^{i}\label{upd:con2}
\end{align}
{\color{black}where $\hat{x}_{0|1}^{i}$ is an estimate of state $x_{0}^{i}$ of the subsystem $i$ calculated at time instant $k=1$; $\hat{w}_{0|1}^{i}$ is an estimate of disturbances $w_{0}^{i}$ of the $i$th subsystem obtained at time instant $k=1$.}
\end{subequations}
The following solution can be obtained for sampling time instant $k=1$ by applying Lagrangian conditions:
\begin{equation}\label{upd:Lagrange:1}
\begin{aligned}
  &\quad\left[\begin{array}{cccccccc}
        P^{-1}_{i,0|-1}   & C_{[:,i]}^{\mathrm{T}} &  &A_{ii}^{\mathrm{T}}  &  &  & \sum\limits_{l\in\mathbb{N}\setminus\{i\}}A_{li}^{\mathrm{T}}C_{[:,l]}^{\mathrm{T}} &  \\
         C_{[:,i]}  &  & I &  &  &  &  &  \\
           & I & R^{-1} &  &  &  &  &  \\
         A_{ii}  &  &  &  & I & -I &  &  \\
           &  &  & I & Q_{i}^{-1} &  &  &  \\
           &  &  & -I &  &  & C_{[:,i]}^{\mathrm{T}} &  \\
          \sum\limits_{l\in\mathbb{N}\setminus\{i\}}C_{[:,l]}A_{li}&  &  &  &  & C_{[:,i]} &  & I \\
           &  &  &  &  &  & I & R^{-1}
        \end{array}\right]\left[\begin{array}{c}
                                  \hat{x}_{0|1}^{i} \\
                                  \lambda_{0}^{i} \\
                                   \hat{\mathbf{v}}_{0|1}^{i}\\
                                   \pi_{0}^{i}\\
                                   \hat{w}_{0|1}^{i}\\
                                   \hat{x}_{1|1}^{i}\\
                                   \lambda_{1}^{i}\\
                                   \hat{\mathbf{v}}_{1|1}^{i}
                                \end{array}\right]\\
                                &=\left[\begin{array}{c}
                                                           P_{i,0|-1}^{-1}\hat{x}_{0|-1}^{i} \\
                                                           y_{0}-\sum\limits_{l\in\mathbb{N}\setminus\{i\}}C_{[:,l]}\hat{x}_{0|-1}^{l} \\
                                                            0\\
                                                            -\sum\limits_{l\in\mathbb{N}\setminus\{i\}}A_{il}\hat{x}_{0|0}^{l}\\
                                                            0\\
                                                            0\\
                                                            y_{1}-\sum\limits_{l\in\mathbb{N}\setminus\{i\}}\sum\limits_{j\in\mathbb{N}\setminus\{i\}}C_{[:,l]}A_{lj}\hat{x}_{0|0}^{j}\\
                                                           0
                                                         \end{array}\right]
                                                         \end{aligned}
\end{equation}
\normalsize
where $\lambda_{0}^{i}$, $\pi_{0}^{i}$, and $\lambda_{1}^{i}$ are the vectors of Lagrange multipliers of subsystem $i$ associated with the equality constraints \eqref{upd:con1}, \eqref{upd:con3}, and \eqref{upd:con2}, respectively.
According to \eqref{upd:Lagrange:1}, it yields
\begin{subequations}\label{upd:eq32}
\begin{align}
  \hat{\mathbf{v}}_{0|1}^{i} & = -C_{[:,i]}\hat{x}_{0|1}^{i}+y_{0}-\sum_{l\in\mathbb{N}\setminus\{i\}}C_{[:,l]}\hat{x}_{0|0}^{l} \\
   \lambda_{0}^{i}& = -R^{-1}\hat{\mathbf{v}}_{0|1}^{i}
\end{align}
\end{subequations}
Denote
\begin{align}
   P_{i,0|0}^{-1}&=P_{i,0|-1}^{-1}+C_{[:,i]}^{\mathrm{T}}R^{-1}C_{[:,i]}  \\
  \hat{x}_{0|0}^{i} & =\hat{x}_{0|-1}^{i}+P_{i,0|0}C_{[:,i]}^{\mathrm{T}}R^{-1}\Big(y_{0}-\sum_{l\in\mathbb{N}\setminus\{i\}}C_{[:,l]}\hat{x}_{0|-1}^{l}-C_{[:,i]}\hat{x}_{0|-1}^{i}\Big)\nonumber
\end{align}
Base on \eqref{upd:eq32}, it holds
\begin{equation}\label{upd:eq:33}
\begin{aligned}
 &\quad P_{i,0|-1}^{-1}\hat{x}_{0|-1}^{i}+C_{[:,i]}^{\mathrm{T}}R^{-1}\big(y_{0}-\sum_{l\in\mathbb{N}\setminus\{i\}}C_{[:,l]}\hat{x}_{0|0}^{l}\big)\\
 &=(P_{i,0|-1}^{-1}+C_{[:,i]}^{\mathrm{T}}R^{-1}C_{[:,i]})\hat{x}_{0|-1}^{i}+C_{[:,i]}^{\mathrm{T}}R^{-1}\big(y_{0}-\sum_{l\in\mathbb{N}\setminus\{i\}}C_{[:,l]}\hat{x}_{0|-1}^{l}-C_{[:,i]}\hat{x}_{0|-1}^{i}\big)\\
 & = P_{i,0|0}^{-1}\hat{x}_{0|0}^{i}
 \end{aligned}
\end{equation}
From \eqref{upd:eq32} and \eqref{upd:eq:33}, we can derive the following
\small
\begin{equation}\label{upd:eq:34}
  \begin{aligned}
& \left[\begin{array}{cccccc}
         P_{i,0|0}^{-1}  & A_{ii}^{\mathrm{T}}  &  &  & \sum\limits_{l\in\mathbb{N}\setminus\{i\}}A_{li}^{\mathrm{T}}C_{[:,l]}^{\mathrm{T}} &  \\
          A_{ii} &  & I & -I &  &  \\
           & I &Q_{i}^{-1}  &  &  &  \\
           & -I &  &  & C_{[:,i]}^{\mathrm{T}} &  \\
          \sum\limits_{l\in\mathbb{N}\setminus\{i\}}C_{[:,l]}A_{li} &  &  & C_{[:,i]} &  & I \\
           &  &  &  & I & R^{-1}
        \end{array}\right]\left[\begin{array}{c}
                                  \hat{x}_{0|1}^{i} \\
                                   \pi_{0}^{i}\\
                                   \hat{w}_{0|1}^{i}\\
                                   \hat{x}_{1|1}^{i}\\
                                   \lambda_{1}^{i}\\
                                   \hat{\mathbf{v}}_{1|1}^{i}
                                \end{array}\right]
                                =\left[\begin{array}{c}
                                                          P_{i,0|0}^{-1}\hat{x}_{0|0}^{i} \\
                                                            -\sum\limits_{l\in\mathbb{N}\setminus\{i\}}A_{il}\hat{x}_{0|0}^{l}\\
                                                            0\\
                                                            0\\
                                                            y_{1}-\sum\limits_{l\in\mathbb{N}\setminus\{i\}}\sum\limits_{j\in\mathbb{N}\setminus\{i\}}C_{[:,l]}A_{lj}\hat{x}_{0|0}^{j}\\
                                                           0
                                                         \end{array}\right]
\end{aligned}
\end{equation}
\normalsize
According to \eqref{upd:eq:34}, it is further obtained
\begin{subequations}\label{upd:eq:35}
  \begin{align}
  \pi_{0}^{i} &= C_{[:,i]}^{\mathrm{T}}\lambda_{1}^{i}\\
  \hat{w}_{0|1}^{i}&=-Q_{i}\pi_{0}^{i}\\
  A_{ii}\hat{x}_{0|1}^{i} &= -\hat{w}_{0|1}^{i}+\hat{x}_{1|1}^{i}-\sum_{l\in\mathbb{N}\setminus\{i\}}A_{il}\hat{x}_{0|0}^{l}
\end{align}
\end{subequations}
Hence, \eqref{upd:eq:34} can be transformed into:
\small
\begin{equation}\label{upd:eq:14}
  \left[\begin{array}{ccc}
         P_{i,0|0}^{-1}  & A_{[:,i]}^{\mathrm{T}}C^{\mathrm{T}} & 0 \\
         CA_{[:,i]}  &-C_{[:,i]}Q_{i}C_{[:,i]}^{\mathrm{T}}  &  I \\
         0  & I & R^{-1}
        \end{array}\right]\left[\begin{array}{c}
                                  \hat{x}_{0|1}^{i} \\
                                  \lambda_{1}^{i} \\
                                  \hat{\mathbf{v}}_{1|1}^{i}
                                \end{array}\right]=\left[\begin{array}{c}
                                                            P_{i,0|0}^{-1}\hat{x}_{0|0}^{i} \\
                                                           y_{1}-\sum\limits_{l\in\mathbb{N}\setminus\{i\}}C_{[:,l]}A_{il}\hat{x}_{0|0}^{l}-\sum\limits_{l\in\mathbb{N}\setminus\{i\}}C_{[:,i]}A_{il}\hat{x}_{0|0}^{l}\\
                                                           0
                                                         \end{array}\right]
\end{equation}
\normalsize
Solving \eqref{upd:eq:14}, we can derive
\begin{subequations}\label{upd:eq:36}
\begin{align}
  \lambda_{1}^{i} &=(CA_{[:,i]}P_{i,0|0}A_{[:,i]}^{\mathrm{T}}C^{\mathrm{T}}+C_{[:,i]}Q_{i}C_{[:,i]}^{\mathrm{T}}+R)^{-1}(CA_{[:,i]}\hat{x}_{0|0}^{i}+\sum_{l=1}\sum_{j\in\mathbb{N}\setminus\{i\}}C_{[:,l]}A_{lj}\hat{x}_{0|0}^{j}-y_{1})  \\
   \hat{x}_{0|1}^{i}&=-P_{i,0|0}A_{[:,i]}^{\mathrm{T}}C^{\mathrm{T}}\lambda_{1}^{i}+\hat{x}_{0|0}^{i}
\end{align}
\end{subequations}
Define
\begin{equation}\label{eq:x10}
  \hat{x}^{i}_{1|0}:=A_{ii}\hat{x}^{i}_{0|0}+\sum_{l\in \mathbb{N}\setminus\{i\}}A_{il}\hat{x}_{0|0}^{l}
\end{equation}
By substituting \eqref{upd:eq:36} into \eqref{upd:eq:35}, it is further obtained that
\begin{equation}\label{upd:eq:37}
\begin{aligned}
  \hat{x}_{1|1}^{i}&=\hat{x}^{i}_{1|0}+(A_{ii}P_{i,0|0}A_{[:,i]}^{\mathrm{T}}C^{\mathrm{T}}+Q_{i}C^{\mathrm{T}}_{[:,i]})(CA_{[:,i]}P_{i,0|0}A_{[:,i]}^{\mathrm{T}}C^{\mathrm{T}}+C_{[:,i]}Q_{i}C_{[:,i]}^{\mathrm{T}}+R)^{-1}\\
  &\quad \times \Big(y_{1}-C_{[:,i]}\hat{x}_{1|0}^{i}-\sum_{l\in\mathbb{N}\setminus\{i\}}C_{[:,l]}\hat{x}_{1|0}^{l}\Big)
\end{aligned}
\end{equation}

Next, let us proceed to the next sampling instant $k=N=2$. The distributed full-information estimation \eqref{fie} becomes
\begin{subequations}\label{upd:fie00}
\begin{align}
&\min_{\hat{x}^{i}_{0|2},~ \hat{w}_{0|2}^{i},~\hat{w}_{1|2}^{i}} \bar{\Phi}_{2}^{i}=\frac{1}{2}\Big(\big\|\hat{x}_{0|2}^{i}-\hat{x}_{0|-1}^{i}\big\|_{P^{-1}_{i,0|-1}}^{2}+\|\hat{\mathbf{v}}_{0|2}^{i}\|^2_{R^{-1}}+\|\hat{w}_{0|2}^{i}\|_{Q_{i}^{-1}}^{2}+\|\hat{\mathbf{v}}_{1|2}^{i}\|^2_{R^{-1}}\nonumber\\
&\quad\quad\quad\quad\quad\quad\quad\quad+\|\hat{w}_{1|2}^{i}\|_{Q_{i}^{-1}}^{2}+\|\hat{\mathbf{v}}_{2|2}^{i}\|^2_{R^{-1}}\Big)\nonumber\\
&\quad\quad\text{s.t.}\quad y_{0} = C_{[:,i]}\hat{x}^{i}_{0|2}+\sum_{l\in\mathbb{N}\setminus\{i\}}C_{[:,l]}\hat{x}_{0|-1}^{l}+\hat{\mathbf{v}}_{0|2}^{i}\label{upd:con12}\\
&\quad\quad\quad~~\hat{x}_{1|2}^{i} = A_{ii}\hat{x}_{0|2}^{i}+\sum_{l\in\mathbb{N}\setminus\{i\}}A_{il}\hat{x}_{0|1}^{l}+\hat{w}_{0|2}^{i}\label{upd:con32}\\
&\quad\quad\quad\quad~ y_{1} = C_{[:,i]}\hat{x}^{i}_{1|2}+\sum_{l\in\mathbb{N}\setminus\{i\}}C_{[:,l]}A_{li}\hat{x}^{i}_{0|2}+\sum_{l\in\mathbb{N}\setminus\{i\}}\sum_{j\in\mathbb{N}\setminus\{i\}}C_{[:,l]}A_{lj}\hat{x}^{j}_{0|0}+\hat{\mathbf{v}}_{1|2}^{i}\label{upd:con22}\\
&\quad\quad\quad~~\hat{x}_{2|2}^{i} = A_{ii}\hat{x}_{1|2}^{i}+\sum_{l\in\mathbb{N}\setminus\{i\}}A_{il}\hat{x}_{1|1}^{l}+\hat{w}_{1|2}^{i}\label{upd:con42}\\
&\quad\quad\quad\quad~ y_{2} = C_{[:,i]}\hat{x}^{i}_{2|2}+\sum_{l\in\mathbb{N}\setminus\{i\}}C_{[:,l]}A_{li}\hat{x}^{i}_{1|2}+\sum_{l\in\mathbb{N}\setminus\{i\}}\sum_{j\in\mathbb{N}\setminus\{i\}}C_{[:,l]}A_{lj}\hat{x}^{j}_{1|1}+\hat{\mathbf{v}}_{2|2}^{i}\label{upd:con52}
\end{align}
\end{subequations}
\normalsize
By formulating Lagrange optimality conditions from \eqref{upd:fie00}, it is obtained that
\small
\begin{align}\label{upd:eq:38}
  &\left[\begin{array}{cccccccccc:ccc}
          P_{i,0|-1}^{-1} & C_{[:,i]}^{\mathrm{T}} &  & A_{ii}^{\mathrm{T}} &  & & \sum\limits_{l\in\mathbb{N}\setminus\{i\}}A_{li}^{\mathrm{T}}C_{[:,i]}^{\mathrm{T}} &  &  &\\
          C_{[:,i]} &  & I &  &  &  &  &  &  &  \\
           & I & R^{-1} &  &  &  &  &  &  &  \\
          A_{ii} &  &  &  & I & -I &  &  &  &   \\
           &  &  & I & Q^{-1}_{i} & &  &  &  &   \\
           &  &  & -I &  &  &  C_{[:,i]}^{\mathrm{T}}&  & A_{ii}^{\mathrm{T}} &   \\
          \sum\limits_{l\in\mathbb{N}\setminus\{i\}}C_{[:,i]}A_{li} &  &  &  &  &  C_{[:,i]}&  & I &  &   \\
           &  &  &  &  &  & I & R^{-1} &  &    \\
           &  &  &  &  & A_{ii} &  &  &  & I   \\
           &  &  &  &  &  &  &  & I & Q_{i}^{-1}  \\
           &  &  &  &  &  &  &  & -I &   \\
           &  &  &  &  &  \sum\limits_{l\in\mathbb{N}\setminus\{i\}}C_{[:,i]}A_{li} & &  &  &   \\
           &  &  &  &  &  &  &  &  &
        \end{array}\right.\nonumber\\
        &\left.\begin{array}{:cccc}
                  &  &  &  \\
                 &  &  & \\
                 &  &  & \\
                 &  &  & \\
                 &  &  & \\
                 &  &\sum\limits_{l\in\mathbb{N}\setminus\{i\}}A_{li}^{\mathrm{T}}C_{[:,i]}^{\mathrm{T}}  & \\
                 &  &  & \\
                 &  &  & \\
                 & -I &  & \\
                 &  &  & \\
                 &  & C_{[:,i]}^{\mathrm{T}} & \\
                 &  C_{[:,i]}&  & I \\
                 &  & I & R^{-1}
               \end{array}\right]\times\left[\begin{array}{c}
                                  \hat{x}_{0|2}^{i} \\
                                  \lambda_{0}^{i} \\
                                   \hat{\mathbf{v}}_{0|2}^{i}\\
                                   \pi_{0}^{i}\\
                                   \hat{w}_{0|2}^{i}\\
                                   \hat{x}_{1|2}^{i}\\
                                   \lambda_{1}^{i}\\
                                   \hat{\mathbf{v}}_{1|2}^{i}\\
                                  \pi_{1}^{i}\\
                                   \hat{w}_{1|2}^{i}\\
                                   \hat{x}_{2|2}^{i}\\
                                   \lambda_{2}^{i}\\
                                   \hat{\mathbf{v}}_{2|2}^{i}
                                \end{array}\right]
                                =\left[\begin{array}{c}
                                  P_{i,0|-1}^{-1}\hat{x}_{0|-1}^{i} \\
                                                           y_{0}-\sum\limits_{l\in\mathbb{N}\setminus\{i\}}C_{[:,l]}\hat{x}_{0|-1}^{l} \\
                                                            0\\
                                                            -\sum\limits_{l\in\mathbb{N}\setminus\{i\}}A_{il}\hat{x}_{0|1}^{l}\\
                                                            0\\
                                                            0\\
                                                            y_{1}-\sum\limits_{l\in\mathbb{N}\setminus\{i\}}\sum\limits_{j\in\mathbb{N}\setminus\{i\}}C_{[:,l]}A_{lj}\hat{x}_{0|0}^{j}\\
                                                           0\\
                                 -\sum\limits_{l\in\mathbb{N}\setminus\{i\}}A_{il}\hat{x}_{1|1}^{l}\\
                                                            0\\
                                                            0\\
                                                            y_{2}-\sum\limits_{l\in\mathbb{N}\setminus\{i\}}\sum\limits_{j\in\mathbb{N}\setminus\{i\}}C_{[:,l]}A_{lj}\hat{x}_{1|1}^{j}\\
                                                           0
                                \end{array}\right]
\end{align}
\normalsize
where $\lambda_{0}^{i}$,  $\pi_{0}^{i}$,  $\lambda_{1}^{i}$, $\pi_{1}^{i}$ and $\lambda_{2}^{i}$ are the vectors of Lagrange multipliers of subsystem $i$ associated with the equality constraints \eqref{upd:con12}, \eqref{upd:con32}, \eqref{upd:con22}, \eqref{upd:con42}, and \eqref{upd:con52}, respectively. Following the derivations of \eqref{upd:eq32} adopted for the previous time instant $k=1$,
\eqref{upd:eq:38} leads to:
\footnotesize
\begin{align}\label{upd:eq:38_2}
  &\left[\begin{array}{ccccccccccccc}
          P_{i,0|0}^{-1} & A_{ii}^{\mathrm{T}} &  & & \sum\limits_{l\in\mathbb{N}\setminus\{i\}}A_{li}^{\mathrm{T}}C_{[:,i]}^{\mathrm{T}} &  &  &  &  &  &  \\
          A_{ii}  &  & I & -I &  &  &  &  &  &  &  \\
           &  I & Q^{-1}_{i} & &  &  &  &  &  &  &  \\
           &  -I &  &  &  C_{[:,i]}^{\mathrm{T}}&  & A_{ii}^{\mathrm{T}} &  &  &\sum\limits_{l\in\mathbb{N}\setminus\{i\}}A_{li}^{\mathrm{T}}C_{[:,i]}^{\mathrm{T}}  & \\
          \sum\limits_{l\in\mathbb{N}\setminus\{i\}}C_{[:,i]}A_{li} &   &  &  C_{[:,i]}&  & I &  &  &  &  & \\
           &   &  &  & I & R^{-1} &  &  &  &  &  \\
           &    &  & A_{ii} &  &  &  & I & -I &  &  \\
           &    &  &  &  &  & I & Q_{i}^{-1} &  &  &  \\
           &    &  &  &  &  & -I &  &  & C_{[:,i]}^{\mathrm{T}} &  \\
           &    &  &  \sum\limits_{l\in\mathbb{N}\setminus\{i\}}C_{[:,i]}A_{li} & &  &  &  &  C_{[:,i]}&  & I \\
           &    &  &  &  &  &  &  &  & I & R^{-1}
        \end{array}\right]\nonumber\\
        &\times\left[\begin{array}{c}
                                  \hat{x}_{0|2}^{i} \\
                                   \pi_{0}^{i}\\
                                   \hat{w}_{0|2}^{i}\\
                                   \hat{x}_{1|2}^{i}\\
                                   \lambda_{1}^{i}\\
                                   \hat{\mathbf{v}}_{1|2}^{i}\\
                                  \pi_{1}^{i}\\
                                   \hat{w}_{1|2}^{i}\\
                                   \hat{x}_{2|2}^{i}\\
                                   \lambda_{2}^{i}\\
                                   \hat{\mathbf{v}}_{2|2}^{i}
                                \end{array}\right]
                                =\left[\begin{array}{c}
                                  P_{i,0|0}^{-1}\hat{x}_{0|0}^{i} \\
                                                            -\sum\limits_{l\in\mathbb{N}\setminus\{i\}}A_{il}\hat{x}_{0|1}^{l}\\
                                                            0\\
                                                            0\\
                                                            y_{1}-\sum\limits_{l\in\mathbb{N}\setminus\{i\}}\sum\limits_{j\in\mathbb{N}\setminus\{i\}}C_{[:,l]}A_{lj}\hat{x}_{0|0}^{j}\\
                                                           0\\
                                 -\sum\limits_{l\in\mathbb{N}\setminus\{i\}}A_{il}\hat{x}_{1|1}^{l}\\
                                                            0\\
                                                            0\\
                                                            y_{2}-\sum\limits_{l\in\mathbb{N}\setminus\{i\}}\sum\limits_{j\in\mathbb{N}\setminus\{i\}}C_{[:,l]}A_{lj}\hat{x}_{1|1}^{j}\\
                                                           0
                                \end{array}\right]
\end{align}
\normalsize
According to \eqref{upd:eq:38}, it is further derived that
\begin{subequations}\label{upd:eq:39}
  \begin{align}
  \hat{x}_{1|2}^{i} &=A_{ii}\hat{x}_{0|2}^{i}+\sum_{l\in\mathbb{N}\setminus\{i\}}A_{il}\hat{x}_{0|1}^{l}+\hat{w}_{0|2}^{i}\label{upd:eq:39_1} \\
     \hat{w}_{0|2}^{i}&=-Q_{i}\pi_{0}^{i}\label{upd:eq:39_2}  \\
     \pi_{0}^{i} & =C_{[:,i]}^{\mathrm{T}}\lambda_{1}^{i}+A_{ii}^{\mathrm{T}}\pi_{1}^{i}+\sum_{l\in\mathbb{N}\setminus\{i\}}A_{li}^{\mathrm{T}}C_{[:,l]}^{\mathrm{T}}\lambda_{2}^{i}\label{upd:eq:39_3}\\
     \lambda_{1}^{i}&=-R^{-1}\hat{\mathbf{v}}_{1|2}^{i}\label{upd:eq:39_4}
  \end{align}
\end{subequations}
and
\begin{subequations}\label{upd:eq:40}
  \begin{align}
    \hat{x}_{0|2}^{i}+P_{i,0|0}A_{ii}^{\mathrm{T}}\pi_{0}^{i}+P_{i,0|0}\sum_{l\in\mathbb{N}\setminus\{i\}}A_{li}^{\mathrm{T}}C_{[:,i]}^{\mathrm{T}}\lambda_{1}^{i}&=\hat{x}_{0|0}^{i}\label{upd:eq:40_1}\\
    \sum_{l\in\mathbb{N}\setminus\{i\}}C_{[:,i]}A_{li}\hat{x}_{0|2}^{i}+ C_{[:,i]}\hat{x}_{1|2}^{i}+\hat{\mathbf{v}}_{1|2}^{i}&=  y_{1}-\sum_{l\in\mathbb{N}\setminus\{i\}}\sum_{j\in\mathbb{N}\setminus\{i\}}C_{[:,l]}A_{lj}\hat{x}_{0|0}^{j}\label{upd:eq:40_2}
  \end{align}
\end{subequations}
By substituting \eqref{upd:eq:39_2} and \eqref{upd:eq:39_3} into \eqref{upd:eq:39_1}, we can derive the following
\begin{equation}\label{upd:eq:41}
  \hat{x}_{1|2}^{i}=A_{ii}\hat{x}_{0|2}^{i}+\sum_{l\in\mathbb{N}\setminus\{i\}}A_{il}\hat{x}_{0|1}^{l}-Q_{i}C_{[:,i]}^{\mathrm{T}}\lambda_{1}^{i}-Q_{i}A_{ii}^{\mathrm{T}}\pi_{1}^{i}-Q_{i}\sum_{l\in\mathbb{N}\setminus\{i\}}A_{li}^{\mathrm{T}}C_{[:,l]}^{\mathrm{T}}\lambda_{2}^{i}
\end{equation}
Considering \eqref{upd:eq:39_3} and \eqref{upd:eq:41}, \eqref{upd:eq:40_1} and \eqref{upd:eq:40_2} are equivalent to the equations
\begin{equation}\label{upd:eq:42_1}
  \hat{x}_{0|2}^{i}=\hat{x}_{0|0}^{i}-P_{i,0|0}A_{[:,i]}^{\mathrm{T}}C^{\mathrm{T}}\lambda_{1}^{i}-P_{i,0|0}(A_{ii}^{\mathrm{T}})^{2}\pi_{1}^{i}-P_{i,0|0}A_{ii}^{\mathrm{T}}\sum_{l\in\mathbb{N}\setminus\{i\}}A_{li}^{\mathrm{T}}C_{[:,l]}^{\mathrm{T}}\lambda_{2}^{i}
\end{equation}
and
  \begin{align}\label{upd:eq:42_2}
  &\quad CA_{[:,i]}\hat{x}_{0|2}^{i}-C_{[:,i]}Q_{i}C_{[:,i]}^{\mathrm{T}}\lambda_{1}^{i}-C_{[:,i]}Q_{i}A_{ii}^{\mathrm{T}}\pi_{1}^{i}-C_{[:,i]}Q_{i}\sum_{l\in\mathbb{N}\setminus\{i\}}A_{li}^{\mathrm{T}}C_{[:,l]}^{\mathrm{T}}\lambda_{2}^{i}+\hat{\mathbf{v}}_{1|2}^{i}\nonumber\\
  &=  y_{1}-\sum_{l\in\mathbb{N}}\sum_{j\in\mathbb{N}\setminus\{i\}}C_{[:,l]}A_{lj}\hat{x}_{0|0}^{j}
  \end{align}
Define $M_{i}=CA_{[:,i]}P_{i,0|0}A_{[:,i]}^{\mathrm{T}}C^{\mathrm{T}}+C_{[:,i]}Q_{i}C_{[:,i]}^{\mathrm{T}}$ and $Z_{i}=CA_{[:,i]}P_{i,0|0}A_{ii}^{\mathrm{T}}+C_{[:,i]}Q_{i}$.
Left multiplying both sides of equation \eqref{upd:eq:42_1} with $CA_{[:,i]}$ and considering \eqref{upd:eq:42_2} gives
\begin{equation}\label{upd:eq:43}
\begin{aligned}
  &\quad -M_{i}\lambda_{1}^{i}-Z_{i}A_{ii}^{\mathrm{T}}\pi_{1}^{i}-Z_{i}\sum_{l\in\mathbb{N}\setminus\{i\}}A_{li}^{\mathrm{T}}C_{[:,l]}^{\mathrm{T}}\lambda_{2}^{i}+\hat{\mathbf{v}}_{1|2}^{i}\\
  &=y_{1}-\sum_{l\in\mathbb{N}}\sum_{j\in\mathbb{N}\setminus\{i\}}C_{[:,l]}A_{lj}\hat{x}_{0|0}^{j}-CA_{[:,i]}\hat{x}_{0|0}^{i}
\end{aligned}
\end{equation}
By taking into account \eqref{upd:eq:39_4} and \eqref{upd:eq:43}, we can obtain
\begin{equation}\label{upd:eq:44}
\begin{aligned}
  &\quad -(M_{i}+R)\lambda_{1}^{i}-Z_{i}A_{ii}^{\mathrm{T}}\pi_{1}^{i}-Z_{i}\sum_{l\in\mathbb{N}\setminus\{i\}}A_{li}^{\mathrm{T}}C_{[:,l]}^{\mathrm{T}}\lambda_{2}^{i}\\
  &=y_{1}-\sum_{l\in\mathbb{N}}\sum_{j\in\mathbb{N}\setminus\{i\}}C_{[:,l]}A_{lj}\hat{x}_{0|0}^{j}-CA_{[:,i]}\hat{x}_{0|0}^{i}\\
  &=y_{1}-C_{[:,i]}\hat{x}_{1|0}^{i}-\sum_{l\in\mathbb{N}\setminus\{i\}}C_{[:,l]}\hat{x}_{1|0}^{l}
\end{aligned}
\end{equation}
Since $M_{i}+R$ is positive-definite, by left-multiplying $(M_{i}+R)^{-1}$ on the both side of \eqref{upd:eq:44}, it holds that:
\begin{equation}\label{upd:eq:46}
\begin{aligned}
  \lambda_{1}^{i}&=-(M_{i}+R)^{-1}Z_{i}A_{ii}^{\mathrm{T}}\pi_{1}^{i}-(M_{i}+R)^{-1}Z_{i}\sum_{l\in\mathbb{N}\setminus\{i\}}A_{li}^{\mathrm{T}}C_{[:,l]}^{\mathrm{T}}\lambda_{2}^{i}\\
  &\quad-(M_{i}+R)^{-1}\Big(y_{1}-C_{[:,i]}\hat{x}_{1|0}^{i}-\sum_{l\in\mathbb{N}\setminus\{i\}}C_{[:,l]}\hat{x}_{1|0}^{l}\Big)
\end{aligned}
\end{equation}
Substituting \eqref{upd:eq:42_1} into \eqref{upd:eq:41} yields
\begin{equation}\label{upd:eq:47}
  \begin{aligned}
  \hat{x}_{1|2}^{i}&=-A_{ii}P_{i,0|0}A_{[:,i]}^{\mathrm{T}}C^{\mathrm{T}}\lambda_{1}^{i}-A_{ii}P_{i,0|0}(A_{ii}^{\mathrm{T}})^{2}\pi_{1}^{i}-A_{ii}P_{i,0|0}A_{ii}^{\mathrm{T}}\sum_{l\in\mathbb{N}\setminus\{i\}}A_{li}^{\mathrm{T}}C_{[:,l]}^{\mathrm{T}}\lambda_{2}^{i}+A_{ii}\hat{x}_{0|0}^{i}\\
  &\quad+\sum_{l\in\mathbb{N}\setminus\{i\}}A_{il}\hat{x}_{0|1}^{l}-Q_{i}C_{[:,i]}^{\mathrm{T}}\lambda_{1}^{i}-Q_{i}A_{ii}^{\mathrm{T}}\pi_{1}^{i}-Q_{i}\sum_{l\in\mathbb{N}\setminus\{i\}}A_{li}^{\mathrm{T}}C_{[:,l]}^{\mathrm{T}}\lambda_{2}^{i}\\
  &=-Z_{i}\lambda_{1}^{i}-(A_{ii}P_{i,0|0}A_{ii}^{\mathrm{T}}+Q_{i})A_{ii}^{\mathrm{T}}\pi_{1}^{i}-(A_{ii}P_{i,0|0}A_{ii}^{\mathrm{T}}+Q_{i})\sum_{l\in\mathbb{N}\setminus\{i\}}A_{li}^{\mathrm{T}}C_{[:,l]}^{\mathrm{T}}\lambda_{2}^{i}+\hat{x}_{1|0}^{i}
  \end{aligned}
\end{equation}
By substituting \eqref{upd:eq:46} into \eqref{upd:eq:47}, we can further obtain that
\begin{equation}\label{upd:eq:48}
  \begin{aligned}
  \hat{x}_{1|2}^{i}
  &=\hat{x}_{1|0}^{i}+\Big(Z_{i}^{\mathrm{T}}(M_{i}+R)^{-1}Z_{i}-(A_{ii}P_{i,0|0}A_{ii}^{\mathrm{T}}+Q_{i})\Big)A_{ii}^{\mathrm{T}}\pi_{1}^{i}\\
  &\quad+\Big(Z_{i}^{\mathrm{T}}(M_{i}+R)^{-1}Z_{i}-(A_{ii}P_{i,0|0}A_{ii}^{\mathrm{T}}+Q_{i})\Big)\sum_{l\in\mathbb{N}\setminus\{i\}}A_{li}^{\mathrm{T}}C_{[:,l]}^{\mathrm{T}}\lambda_{2}^{i}\\
  &\quad+Z_{i}^{\mathrm{T}}(M_{i}+R)^{-1}\Big(y_{1}-C_{[:,i]}\hat{x}_{1|0}^{i}-\sum_{l\in\mathbb{N}\setminus\{i\}}C_{[:,l]}\hat{x}_{1|0}^{l}\Big)
  \end{aligned}
\end{equation}
From \eqref{upd:eq:37}, we have $\hat{x}_{1|1}^{i} = \hat{x}_{1|0}^{i}+(Z_{i})^{\mathrm{T}}(M_{i}+R)^{-1}(y_{1}-C_{[:,i]}\hat{x}_{1|0}^{i}-\sum_{l\in\mathbb{N}\setminus\{i\}}C_{[:,l]}\hat{x}_{1|0}^{i})$. Hence, \eqref{upd:eq:48} can be rewritten as
\begin{equation}\label{upd:eq:49}
\begin{aligned}
  \hat{x}_{1|2}^{i}&=\hat{x}_{1|1}^{i}+\Big(Z_{i}^{\mathrm{T}}(M_{i}+R)^{-1}Z_{i}-(A_{ii}P_{i,0|0}A_{ii}^{\mathrm{T}}+Q_{i})\Big)A_{ii}^{\mathrm{T}}\pi_{1}^{i}\\
  &\quad+\Big(Z_{i}^{\mathrm{T}}(M^{i}+R)^{-1}Z_{i}-(A_{ii}P_{i,0|0}A_{ii}^{\mathrm{T}}+Q_{i})\Big)\sum_{l\in\mathbb{N}\setminus\{i\}}A_{li}^{\mathrm{T}}C_{[:,l]}^{\mathrm{T}}\lambda_{2}^{i}\\
\end{aligned}
\end{equation}
Define $P_{i,1|1}:=-Z_{i}^{\mathrm{T}}(M_{i}+R)^{-1}Z_{i}+A_{ii}P_{i,0|0}A_{ii}^{\mathrm{T}}+Q_{i}$, and \eqref{upd:eq:49} can be expressed as:
\begin{equation}\label{upd:eq:49_2}
  P_{i,1|1}^{-1}\hat{x}_{1|2}^{i}+A_{ii}^{\mathrm{T}}\pi_{1}^{i}+\sum_{l\in\mathbb{N}\setminus\{i\}}A_{li}^{\mathrm{T}}C_{[:,l]}^{\mathrm{T}}\lambda_{2}^{i}=P_{i,1|1}^{-1}\hat{x}_{1|1}^{i}\\
\end{equation}
Therefore, \eqref{upd:eq:38} is equivalent to
\small
\begin{equation}\label{upd:eq:50}
  \begin{aligned}
& \left[\begin{array}{cccccc}
         P_{i,1|1}^{-1}  & A_{ii}^{\mathrm{T}}  &  &  & \sum\limits_{l\in\mathbb{N}\setminus\{i\}}A_{li}^{\mathrm{T}}C_{[:,l]}^{\mathrm{T}} &  \\
          A_{ii} &  & I & -I &  &  \\
           & I &Q_{i}^{-1}  &  &  &  \\
           & -I &  &  & C_{[:,i]}^{\mathrm{T}} &  \\
          \sum\limits_{l\in\mathbb{N}\setminus\{i\}}C_{[:,l]}A_{li} &  &  & C_{[:,i]} &  & I \\
           &  &  &  & I & R^{-1}
        \end{array}\right]\left[\begin{array}{c}
                                  \hat{x}_{1|2}^{i} \\
                                   \pi_{1}^{i}\\
                                   \hat{w}_{1|2}^{i}\\
                                   \hat{x}_{2|2}^{i}\\
                                   \lambda_{2}^{i}\\
                                   \hat{\mathbf{v}}_{2|2}^{i}
                                \end{array}\right]
                                =\left[\begin{array}{c}
                                                          P_{i,1|1}^{-1}\hat{x}_{1|1}^{i} \\
                                                            -\sum\limits_{l\in\mathbb{N}\setminus\{i\}}A_{il}\hat{x}_{1|1}^{l}\\
                                                            0\\
                                                            0\\
                                                            y_{2}-\sum\limits_{l\in\mathbb{N}\setminus\{i\}}\sum\limits_{j\in\mathbb{N}\setminus\{i\}}C_{[:,l]}A_{lj}\hat{x}_{1|1}^{j}\\
                                                           0
                                                         \end{array}\right]
\end{aligned}
\end{equation}
\normalsize
which has a similar structure as \eqref{upd:eq:34}. Similarly, if we define $\hat{x}^{i}_{2|1}:=A_{ii}\hat{x}^{i}_{1|1}+\sum_{l\in \mathbb{N}\setminus\{i\}}A_{il}\hat{x}_{1|1}^{l}$, it can be obtained that
\begin{equation}\label{upd:eq:51}
\begin{aligned}
  \hat{x}_{2|2}^{i}&=\hat{x}^{i}_{2|1}+(A_{ii}P_{i,1|1}A_{[:,i]}^{\mathrm{T}}C^{\mathrm{T}}+Q_{i}C^{\mathrm{T}}_{[:,i]})(CA_{[:,i]}P_{i,1|1}A_{[:,i]}^{\mathrm{T}}C^{\mathrm{T}}+C_{[:,i]}Q_{i}C_{[:,i]}^{\mathrm{T}}+R)^{-1}\\
  &\quad \times \Big(y_{2}-C_{[:,i]}\hat{x}_{2|1}^{i}-\sum_{l\in\mathbb{N}\setminus\{i\}}C_{[:,l]}\hat{x}_{2|1}^{l}\Big)
\end{aligned}
\end{equation}

{\color{black}In the following, we can iteratively apply the same methodology to derive the analytical solution for $k=3$, $k=4$, and subsequent time instants. The induction method allows us to derive the analytical solution for the optimization problem in a recursive fashion.
Specifically, at sampling instant $k$,
the recursive expression of the Kalman filter for the $i$th subsystem of the proposed distributed scheme is presented as follows:}

\begin{subequations}\label{upd:dkf}
\begin{itemize}
  \item[] \textit{Prediction:}
  \begin{equation}\label{upd:pre:x}
    \hat{x}_{k|k-1}^{i}=A_{ii}\hat{x}_{k-1|k-1}^{i}+\sum_{l\in\mathbb{N}\setminus\{i\}}A_{il}\hat{x}_{k-1|k-1}^{l}
  \end{equation}
  \item[]\textit{Update:}
  \begin{equation}\label{upd:upd:x}
    \hat{x}_{k|k}^{i}=\hat{x}_{k|k-1}^{i}+L_{i,k}\big(y_{k}-C_{[:,i]}\hat{x}_{k|k-1}^{i}-\sum_{l\in\mathbb{N}\setminus\{i\}}C_{[:,l]}\hat{x}_{k|k-1}^{l}\big)
  \end{equation}
  \item[] \textit{Kalman filter gain:}
  \begin{equation}
    L_{i,k}=(CA_{[:,i]}P_{i,k-1|k-1}A_{ii}^{\mathrm{T}}+C_{[:,i]}Q_{i})^{\mathrm{T}}(CA_{[:,i]}P_{i,k-1|k-1}A_{[:,i]}^{\mathrm{T}}C^{\mathrm{T}}+C_{[:,i]}Q_{i}C_{[:,i]}^{\mathrm{T}}+R)^{-1}\label{upd:eq:26d}
  \end{equation}
    \item[] \textit{Covariance matrix:}
  \begin{align}
P_{i,k|k}&=-L_{i,k}(CA_{[:,i]}P_{i,k-1|k-1}A_{ii}^{\mathrm{T}}+C_{[:,i]}Q_{i})+(A_{ii}P_{i,k-1|k-1}A_{ii}^{\mathrm{T}}+Q_{i})\label{upd:p}
  \end{align}
\end{itemize}
\end{subequations}

{\color{black}The recursive expression of the proposed DKF approach is derived by leveraging the distributed full-information estimation formulation in \eqref{fie}.
Therefore, from a theoretical point of view, the estimates provided by the proposed DKF and the distributed full-information estimation are identical, in an unconstrained setting.}
\section*{Proposed distributed extended Kalman filter}
As has been discussed, our objective is to propose a distributed estimation method based on the subsystem models in the form of \eqref{model_nonlinear}, to estimate the full-state of the nonlinear process in \eqref{model_nonlinear_centralized}. In this section, the proposed linear DKF algorithm in \eqref{upd:dkf} is extended to formulate a distributed extended Kalman filter scheme for nonlinear state estimation within a distributed framework.

Instead of using linear subsystem models, in this nonlinear design, the nonlinear subsystem models in the form of \eqref{model_nonlinear} are utilized. By incorporating successive linearization of each subsystem model in \eqref{model_nonlinear} at every new sampling instant $k$, a distributed extended Kalman filter method is formulated based on \eqref{upd:dkf}. The expression of each local filter of the distributed scheme for subsystem $i$, $i\in\mathbb N$, is presented as follows:
\begin{subequations}\label{edkf}
\begin{itemize}
\item[] \textit{Prediction:}
\begin{equation}\label{ekf1}
  \hat{x}_{k|k-1}^{i}= f_{i}(\hat{x}^{i}_{k-1|k-1}, \hat{X}^{i}_{k-1|k-1})
\end{equation}
\item[] \textit{Update:}
\begin{equation}\label{ekf3}
  \hat{x}_{k|k}^{i}=\hat{x}_{k|k-1}^{i}+L_{i,k}(y_{k}-h(\hat{x}_{k|k-1}))
\end{equation}
\item[]  \textit{Kalman filter gain:}
  \begin{align}
    L_{i,k}&=(C_{k}A_{[:,i],k-1}P_{i,k-1|k-1}A_{ii,k-1}^{\mathrm{T}}+C_{[:,i],k}Q_{i})^{\mathrm{T}}\nonumber\\
&\quad\times(C_{k}A_{[:,i],k-1}P_{i,k-1|k-1}A_{[:,i],k-1}^{\mathrm{T}}C_{k}^{\mathrm{T}}+C_{[:,i],k}Q_{i}C_{[:,i],k}^{\mathrm{T}}+R)^{-1}\label{ekf4}
  \end{align}
  \item[]  \textit{Covariance matrix:}
  \begin{equation}
   P_{i,k|k}=-L_{i,k}(C_{k}A_{[:,i],k-1}P_{i,k-1|k-1}A_{ii,k-1}^{\mathrm{T}}+C_{[:,i],k}Q_{i})+A_{ii,k-1}P_{i,k-1|k-1}A_{ii,k-1}^{\mathrm{T}}+Q_{i}\label{ekf5}
  \end{equation}
  \item[]  \textit{Linearization: }
  \begin{align}
    C_{[:,i],k} &=\frac{\partial h(\hat{x}_{k|k-1})}{\partial \hat{x}^{i}_{k|k-1}},\quad\quad\quad~\,\, C_{k}=[C_{[:,1],k},C_{[:,2],k},\ldots,C_{[:,n],k}]\label{ekfcc}\\
    A_{li,k} &= \frac{\partial f_{l}(\hat{x}^{l}_{k|k}, \hat{X}^{l}_{k|k})}{\partial \hat{x}^{i}_{k|k}}, ~\,\, A_{[:,i],k} = \mathrm{col}\{A_{1i,k},A_{2i,k},\ldots,A_{ni,k}\}\label{ekfaa}
  \end{align}
\end{itemize}
\end{subequations}

{\color{black}As an analog of \eqref{upd:dkf}, the design in \eqref{edkf} also involves two steps: the prediction step in which the state dynamic equation of the corresponding subsystem (i.e., \eqref{model_nonlinear1}) is linearized at the current state estimates $\hat{x}_{k|k}^{i}$, and the update step in which the output measurement equation of the subsystem (i.e., \eqref{model_nonlinear_centralized2}) is successively linearized at the current state estimates $\hat{x}_{k|k-1}$.
The following algorithm outlines the key implementation steps for the proposed DEKF.}

\begin{algorithm}
\caption{Execution of the proposed partition-based distributed extended Kalman filter algorithm (DEKF)}\vspace{3mm}
\label{alg1}
~ At initial time instant $k=0$, each EKF-based estimator $i$, $i\in\mathbb{N}$, is initialized with $\hat{x}^{i}_{0|0}$ and $P_{i,0|0}^{-1}$. At each sampling time instant $k = 0, 1, \ldots$, EKF $i$, $i\in\mathbb{N}$, conducts the following steps:
\begin{enumerate}
\item[1.] Prediction
\begin{enumerate}
  \item[1.1.] Receive the most current estimates $\hat{x}_{k|k}^{l}$ from each interacting estimator $l$, $l\in\mathbb{N}\setminus\{i\}$.
  \item[1.2.]  Compute state estimates $\hat{x}_{k+1|k}^{i}$ following \eqref{ekf1}.
\end{enumerate}
\item[2.] Update
\begin{enumerate}
  \item[2.1.] Receive prediction $\hat{x}_{k+1|k}^{l}$ from each interacting subsystem $l$, $l\in\mathbb{N}\setminus\{i\}$.
  \item[2.2.] Receive measurement $y_{k+1}^{i}$ from each subsystem $i$, $i\in\mathbb{N}$.
  \item[2.3.] Compute the \emph{posteriori} covariance matrix $P_{i,k+1|k+1}$ and Kalman filter gain $L_{i,k}$ according to \eqref{ekf4}, \eqref{ekf5}, \eqref{ekfcc}, and \eqref{ekfaa}.
  \item[2.4.] Compute state estimates $\hat{x}_{k+1|k+1}^{i}$ following \eqref{ekf3}.
\end{enumerate}
  \item[3.] Set $k=k+1$. Go to step 1.
\end{enumerate}
\end{algorithm}
\section*{Error dynamics for the proposed DEKF}
In this section, we study the estimation error dynamics for the proposed DEKF method in \eqref{edkf}. Before proceeding further, we present an assumption as follows:
\begin{assumption}\label{assume:continuous}
  The vector-valued functions $f$ and $h$ in \eqref{model_nonlinear_centralized} are twice continuously differentiable on $\mathbb{R}^{n_{x}}$.
\end{assumption}

\begin{proposition}\label{props:3}
Let $\hat{x}_{k|k}=\mathrm{col}\{\hat{x}_{k|k}^{1}, \hat{x}_{k|k}^{2}, \ldots, \hat{x}_{k|k}^{n}\}$. If the Assumption \ref{assume:continuous} holds, then the estimation error $e_{k|k}=x_{k}-\hat{x}_{k|k}$, $k\geq 0$, for the entire process \eqref{model_nonlinear_centralized} given by the proposed DEKF algorithm is governed by the recursion below:
\begin{subequations}\label{eq:rs}
\begin{equation}\label{eq:39}
e_{k|k}=(I-L_{k}C_{k})A_{k-1}e_{k-1|k-1}+r_{k}+s_{k}
\end{equation}
{\text where}
\begin{align}
r_{k}&=(I-L_{k}C_{k})\phi(x_{k-1},\hat{x}_{k-1|k-1})-L_{k}\varphi(x_{k}, \hat{x}_{k|k-1})\label{phi_varphi}\\
s_{k}&=(I-L_{k}C_{k})w_{k-1}-L_{k}v_{k}\label{sk}
\end{align}
\end{subequations}
In \eqref{phi_varphi}$, \phi$ and $\varphi$ are high-order terms of the first-order Taylor-series expansion of $f$ and $h$, respectively, that is, $\phi(x_{k}, \hat{x}_{k|k}) = f(x_{k})-f(\hat{x}_{k|k}) - A_{k}(x_{k}-\hat{x}_{k|k})$ and $\varphi(x_{k}, \hat{x}_{k|k-1})=h(x_{k})-h(\hat{x}_{k|k-1}) -C_{k}(x_{k}-\hat{x}_{k|k-1})$.
\end{proposition}
\begin{proof}
Under Assumption \ref{assume:continuous}, the functions $f$ in \eqref{model_nonlinear_centralized1} and $h$ in \eqref{model_nonlinear_centralized2} can be expanded through Taylor-series expansions as:
\begin{subequations}
  \begin{equation}\label{eq:40}
  f(x_{k})=f(\hat{x}_{k|k}) +A_{k}(x_{k}-\hat{x}_{k|k})+\phi(x_{k}, \hat{x}_{k|k})
\end{equation}
and
\begin{equation}\label{eq:41}
h(x_{k})=h(\hat{x}_{k|k-1}) +C_{k}(x_{k}-\hat{x}_{k|k-1})+\varphi(x_{k}, \hat{x}_{k|k-1})
\end{equation}
\end{subequations}
$A_{k}$ and $C_{k}$ are Jacobian matrices that can be calculated by
\begin{equation*}
  A_{k}=\frac{\partial f(\hat{x}_{k|k})}{\partial x_{k}}
\end{equation*}
and
\begin{equation*}
 C_{k}=\frac{\partial h(\hat{x}_{k|k-1})}{\partial x_{k}};
\end{equation*}
Based on \eqref{model_nonlinear_centralized1}, \eqref{ekf1}, and \eqref{eq:40}, we can derive the one-step-ahead prediction error $e_{k|k-1} =x_{k}-\hat{x}_{k|k-1}$ as:
\begin{equation}\label{eq:k|k-1}
\begin{aligned}
e_{k|k-1} &=f(x_{k-1})+w_{k-1}-f(\hat{x}_{k-1|k-1})\\
&=A_{k-1}e_{k-1|k-1}+\phi(x_{k-1},\hat{x}_{k-1|k-1})+w_{k-1}
\end{aligned}
\end{equation}
where $\hat{x}_{k|k-1}$ is the concatenated vector of the prediction at time instant $k$ given by the proposed DEKF, which is derived by $\hat{x}_{k|k-1}=\mathrm{col}\{\hat{x}_{k|k-1}^{1},\hat{x}_{k|k-1}^{2},\ldots,\hat{x}_{k|k-1}^{n}\}$.
Similarly, based on \eqref{model_nonlinear_centralized2}, \eqref{ekf3}, and \eqref{eq:41}, the estimation error $e_{k|k} =x_{k}-\hat{x}_{k|k}$ can be computed as
\begin{equation}\label{eq:k|k:1}
\begin{aligned}
e_{k|k}
&=x_{k}-\hat{x}_{k|k-1}-L_{k}\big(h(x_{k})+v_{k}-h(\hat{x}_{k|k-1})\big)\\
&=(I-L_{k}C_{k})e_{k|k-1}-L_{k}\varphi(x_{k}, \hat{x}_{k|k-1})-L_{k}v_{k}
\end{aligned}
\end{equation}
where $L_{k} = \mathrm{col}\{L_{k}^{1},L_{k}^{2}, \ldots, L_{k}^{n}\}$. 
By substituting \eqref{eq:k|k-1} into \eqref{eq:k|k:1}, the estimation error dynamics can be described as follows:
\begin{equation}\label{eq:k|k}
\begin{aligned}
e_{k|k}&=(I-L_{k}C_{k})A_{k-1}e_{k-1|k-1}+(I-L_{k}C_{k})\phi(x_{k-1},\hat{x}_{k-1|k-1})\\
&\quad+(I-L_{k}C_{k})w_{k-1}-L_{k}\varphi(x_{k}, \hat{x}_{k|k-1})-L_{k}v_{k}
\end{aligned}
\end{equation}
Then, \eqref{eq:39} is obtained with \eqref{phi_varphi} and \eqref{sk}. $\square$
\end{proof}

\section*{Stability analysis}
In this section, we prove the stability of the proposed partition-based DEKF in \eqref{edkf} for nonlinear processes in \eqref{model_nonlinear_centralized}. At first, we recall some results on the input-to-state stability.
{\color{black}\begin{definition}(Jiang and Wang \cite{jiang2001input})
  Consider a nonlinear system in the following form:
  \begin{equation}\label{iss}
    z_{k+1}=F(z_{k},w_{k})
  \end{equation}
  where $z_{k}$ denotes the state of the system, and $w_{k}$ represents a bounded disturbance of the system such that $\|w_{k}\|\leq \bar{w}$, $\forall k$. System \eqref{iss} is said to be input-to-state stable (ISS), if there is a $\mathcal{KL}$-function $\beta$ and a $\mathcal{K}_{\infty}$-function $\gamma$ such that, for each $z_{0}\in\mathbb{R}^{+}$, it holds
  \begin{equation*}
    \|z_{k}\|\leq\beta(z_{0},k)+\gamma(\bar{w}), ~k=1,2,\ldots
  \end{equation*}
\end{definition}}

\begin{definition}\label{lemma:iss}(Jiang and Wang \cite{jiang2001input})
  A continuous function $V:\mathbb{R}^{n}\rightarrow\mathbb{R}^{+}$ is an ISS-Lyapunov function for system \eqref{iss}, if there exist $\mathcal{K}_{\infty}$-functions $\alpha_{1}$, $\alpha_{2}$, and $\alpha_{3}$, and a $\mathcal{K}$-function $\sigma$ such that
  \begin{subequations}\label{eq:iss}
    \begin{align}
  \alpha_{1}(\|z_{k}\|)&\leq V(z_{k})\leq\alpha_{2}(\|z_{k}\|)\\
  V(F(z_{k},w_{k}))-V(z_{k})&\leq-\alpha_{3}(\|z_{k}\|)+\sigma(\|w_{k}\|)
  \end{align}
  \end{subequations}

\end{definition}
\begin{lemma}(Jiang and Wang \cite{jiang2001input})
  If system \eqref{iss} admits an ISS-Lyapunov function, then it is ISS.
\end{lemma}

Before proceeding further, we introduce Lemma \ref{Matrix Inversion Lemma} and Assumption \ref{assume:1} which are needed for deriving Proposition \ref{upd:props:2}, which will be used to prove the stability of the proposed DEKF approach.

\begin{lemma} (Henderson and Searle \cite{henderson1981deriving})\label{Matrix Inversion Lemma}
For nonsingular matrices $A$ and $D$, it holds that
\begin{equation*}
(A-BD^{-1}C)^{-1} = A^{-1}+A^{-1}B(D-CA^{-1}B)^{-1}CA^{-1}
\end{equation*}
\end{lemma}

\begin{assumption}\label{assume:1}(Reif et al \cite{reif1999stochastic})
There exist positive real numbers $\bar{a}$, $\bar{c}$, $\bar{p}$, $\underline{p}$, $\bar{q}$, $\underline{q}$, $\bar{r}$, and $\underline{r}$, such that the following inequalities characterizing the boundedness of matrices associated with the proposed DEKF hold at each sampling instant $k$, $k\geq 0$:
\begin{align*}
&\underline{a}\leq\|A_{ii,k}\|\leq \bar{a},~\quad\quad\|A_{ij,k}\|\leq \bar{a},
 ~\quad\quad\quad~\,\underline{c}\leq\|C_{[:,i],k}\|\leq \bar{c},\\
&\underline{q}I_{n_{x^{i}}}\leq Q_{i}\leq \bar{q}I_{n_{x^{i}}}, \quad\underline{r}I_{n_{y}}\leq R\leq \bar{r}I_{n_{y}},\,~\quad\underline{p}I_{n_{x^{i}}}\leq P_{i,k|k}\leq \bar{p}I_{n_{x^{i}}}.
\end{align*}
\end{assumption}

For notational simplicity, we define $F_{ij,k}=A_{ij,k-1}-L_{i,k}C_{k}A_{[:,j],k-1}$, $F_{k}=[F_{ij,k}]=(I-L_{k}C_{k})A_{k-1}$, the block-diagonal matrix $F_{d,k}=\mathrm{diag}\{F_{11,k},F_{22,k},\ldots,F_{nn,k}\}$, and the off-diagonal matrix $F_{o,k}=F-F_{d,k}$; $H_{ii,k}=I-L_{i,k}C_{[:,i],k}$.

\begin{proposition}\label{upd:props:2}
Consider positive definite matrix $P_{i,k|k}$ in \eqref{ekf5}, $i\in\mathbb{N}$.
Define $\Pi_{i,k|k}=P_{i,k|k}^{-1}$, and assume that $F_{ii,k}$ is invertable for all $i\in\mathbb{N}$. Then, there is a real number $0<\alpha<1$ such that
\begin{equation}\label{I-LC}
  F_{ii,k}\Pi_{i,k|k}F_{ii,k}^{\mathrm{T}}\leq(1-\alpha)\Pi_{i,k-1|k-1}
\end{equation}
\end{proposition}
\begin{proof}
$P_{i,k|k}$ in \eqref{ekf5} can be rewritten as the equation:
\begin{align*}\label{eq:props1}
  P_{i,k|k}&=(A_{ii,k-1}-L_{i,k}C_{k}A_{[:,i],k-1})P_{i,k-1|k-1}(A_{ii,k-1}-L_{i,k}C_{k}A_{[:,i],k-1})^{\mathrm{T}}\nonumber\\
  &\quad+(I-L_{i,k}C_{[:,i],k})Q_{i}(I-L_{i,k}C_{[:,i],k})^{\mathrm{T}}+L_{i,k}RL_{i,k}^{\mathrm{T}}
\end{align*}
According to the notation above, we obtain
\begin{align}
  P_{i,k|k}&=F_{ii,k}P_{i,k-1|k-1}F_{ii,k}^{\mathrm{T}}+H_{ii,k}Q_{i}H_{ii,k}^{\mathrm{T}}+L_{i,k}RL_{i,k}^{\mathrm{T}}
\end{align}
{\color{black}Since $Q_{i}$ is positive definite, it holds that $H_{ii,k}Q_{i}H_{ii,k}^{\mathrm{T}}>0$. Therefore, the following inequality is satisfied:}
\begin{align}\label{upd:props4:1}
  P_{i,k|k} &\geq F_{ii,k}P_{i,k-1|k-1}F_{ii,k}^{\mathrm{T}}+L_{i,k}RL_{i,k}^{\mathrm{T}}\nonumber\\
  &\geq F_{ii,k}\big(P_{i,k-1|k-1}+F_{ii,k}^{-1}L_{i,k}RL_{i,k}^{\mathrm{T}}(F_{ii,k}^{\mathrm{T}})^{-1}\big)F_{ii,k}^{\mathrm{T}}
\end{align}
{\color{black}From Assumption \ref{assume:1}, we have
\begin{align*}
 \underline{a}\leq\|A_{[:,i],k}\|\leq\sqrt{n}\bar{a}, \quad  \underline{c}\leq\|C_{k}\|\leq\sqrt{n}\bar{c},\quad  \underline{l}\leq\|L_{i,k}\|\leq\bar{l},\quad \|F_{ii,k}\|\leq\bar{f}
\end{align*}
where $\underline{l}=\frac{\underline{c}\underline{a}^{2}\underline{p}+\underline{c}\underline{q}}{(n\bar{c}\bar{a})^{2}\bar{p}+\bar{c}^{2}\bar{q}+\bar{r}}$,  $\bar{l}=\frac{n\bar{c}\bar{a}^{2}\bar{p}+\bar{c}\bar{q}}{(\underline{c}\underline{a})^{2}\underline{p}+\underline{c}^{2}\underline{q}+\underline{r}}$, and $\bar{f}=\sqrt{n}\bar{a}+\sqrt{n}\bar{l}\bar{c}\bar{a}$.
Further, we obtain that
\begin{align}\label{upd:props4:1_1}
  P_{i,k|k}&\geq F_{ii,k}\big(P_{i,k-1|k-1}+\frac{\underline{l}^{2}\underline{r}}{\bar{f}^{2}}I\big)F_{ii,k}^{\mathrm{T}}\nonumber\\
  &\geq \big(1+\frac{\underline{l}^{2}\underline{r}}{\bar{p}\bar{f}^{2}}\big)F_{ii,k} P_{i,k-1|k-1}F_{ii,k}^{\mathrm{T}}
\end{align}
By taking the inverse of the both sides of \eqref{upd:props4:1_1}, and multiplying \eqref{upd:props4:1_1} from left and right with $F_{ii,k}^{\mathrm{T}}$ and $F_{ii,k}$, it further derived
 \begin{equation}\label{upd:props4:1_2}
  F_{ii,k}^{\mathrm{T}}\Pi_{i,k|k}F_{ii,k}\leq \big(1+\frac{\underline{l}^{2}\underline{r}}{\bar{p}\bar{f}^{2}}\big)^{-1}\Pi_{i,k-1|k-1}
 \end{equation}
Define $\alpha = \frac{(\underline{l}^{2}\underline{r})/(\bar{p}\bar{f}^{2})}{1+(\underline{l}^{2}\underline{r})/(\bar{p}\bar{f}^{2})}$, and (48) is yield with 0<alpha<1.}
\end{proof}

Before proceeding further, additional assumptions and Lemmas below will be used for proving Theorem \ref{theorem}.

\begin{assumption}\label{assume:2}
There exist $\epsilon_{\phi}, \epsilon_{\varphi},\delta_{\phi},\delta_{\varphi}>0$, such that $\phi$ and $\varphi$ in \eqref{eq:40} and \eqref{eq:41} satisfy
\begin{equation*}
\|\phi(x_{k}, \hat{x}_{k|k})\|\leq\epsilon_{\phi}\|x_{k}-\hat{x}_{k|k}\|^{2}
\end{equation*}
\text{and}
\begin{equation*}
  \|\varphi(x_{k}, \hat{x}_{k|k-1})\|\leq\epsilon_{\varphi}\|x_{k}-\hat{x}_{k|k-1}\|^{2}
\end{equation*}
for $\|x_{k}-\hat{x}_{k|k}\|\leq \delta_{\phi}$ and $\|x_{k}-\hat{x}_{k|k-1}\|\leq \delta_{\varphi}$, respectively.
\end{assumption}


\begin{assumption}\label{weakly}
The matrices $F_{d,k}$ and $F_{o,k}$ satisfy the following inequality:
\begin{align*}
  &F_{o,k}^{\mathrm{T}}\Pi_{k|k}F_{o,k}+ F_{o,k}^{\mathrm{T}}\Pi_{k|k}F_{d,k}+ F_{d,k}^{\mathrm{T}}\Pi_{k|k}F_{o,k}<\frac{1}{2}T_{k}
\end{align*}
where
  \begin{align*}
  T_{ii,k}&=\Pi_{i,k-1|k-1}F_{ii,k}^{-1}L_{i,k}\Big(R^{-1}+L_{i,k}^{\mathrm{T}}(F_{ii,k}^{\mathrm{T}})^{-1}\Pi_{i,k-1|k-1}F_{ii,k}^{-1}L_{i,k}\Big)^{-1}L_{i,k}^{\mathrm{T}}(F_{ii,k}^{-1})^{\mathrm{T}}\Pi_{i,k-1|k-1}  \\
   T_{k}&=\mathrm{diag}\{T_{11,k}, T_{22,k},\ldots, T_{nn,k}\}
\end{align*}
\end{assumption}
\normalsize
\begin{assumption}\label{ass:bound}
Stochastic disturbances $w_{k}$ and measurement noise $v_{k}$ are bounded, i.e.,
\begin{equation*}
  \|w_{k}\|\leq\bar{w}, \quad \|v_{k}\|\leq\bar{v}.
\end{equation*}
\end{assumption}
{\color{black}
\begin{lemma}\label{lemma1}(Rapp and Nyman \cite{rapp2004stability})
  If Assumptions \ref{assume:1} and \ref{ass:bound} hold, then there exist positive real number $\varepsilon, \kappa_{1}, \kappa_{2}>0$, such that
  \begin{equation*}
   r_{k}^{\mathrm{T}}\Pi_{k|k}\big(2F_{k}e_{k-1|k-1}+r_{k}\big)\leq \kappa_{1}\|e_{k-1|k-1}\|^{3}+\bar{w}\kappa_{2}
  \end{equation*}
  holds for $\|e_{k-1|k-1}\|\leq\varepsilon$.
\end{lemma}

\begin{lemma}\label{lemma2}(Rapp and Nyman \cite{rapp2004stability})
  If Assumptions \ref{assume:1} and \ref{ass:bound} hold, then there exist positive real number $\kappa_{3}>0$ such that
  \begin{equation*}
    2s_{k}^{\mathrm{T}}\Pi_{k|k}\big(F_{k}e_{k-1|k-1}+r_{k}\big)+\|s_{k}\|_{\Pi_{k|k}}^{2}\leq\big((1+\sqrt{n}\bar{l}\bar{c})\bar{w}+\bar{l}\bar{v}\big)\kappa_{3}
  \end{equation*}
   is satisfied for $\|e_{k-1|k-1}\|\leq\varepsilon$.
\end{lemma}}

\begin{theorem}\label{theorem}
Consider a distributed estimation scheme developed following \eqref{edkf} based on the subsystem models in \eqref{model_nonlinear}. If
Assumptions \ref{assume:continuous}-\ref{ass:bound} hold, if there is a bound such that $\|e_{0|0}\|<\delta_{0}$, the estimation error $e_{k|k} = x_k - \hat{x}_{k|k}$ for the overall process \eqref{model_nonlinear_centralized} is ISS.
\end{theorem}
\begin{proof}\label{eq:V45}
Consider a quadratic Lyapunov function for the estimation error for the entire system:
\begin{equation}\label{eq:lyapunov}
  V_{k}(e_{k|k})=\|e_{k|k}\|^{2}_{\Pi_{k|k}}
\end{equation}
which satisfies
\begin{equation*}\label{eq:iss_1}
  \frac{1}{\bar{p}}\|e_{k|k}\|^{2}\leq V_{k}(e_{k|k})\leq\frac{1}{\underline{p}}\|e_{k|k}\|^{2}
\end{equation*}
Substituting \eqref{eq:39} into \eqref{eq:lyapunov}, it yields:
\begin{align}\label{eq:49}
V_{k}(e_{k|k}) &=\|(I-L_{k}C_{k})A_{k-1}e_{k-1|k-1}+r_{k}+s_{k}\|^{2}_{\Pi_{k|k}}\nonumber\\
&= \|(I-L_{k}C_{k})A_{k-1}e_{k-1|k-1}\|^{2}_{\Pi_{k|k}}+r_{k}^{\mathrm{T}}\Pi_{k|k}\big(2(I-L_{k}C_{k})A_{k-1}e_{k-1|k-1}+r_{k}\big)\nonumber\\
&\quad+2s_{k}^{\mathrm{T}}\Pi_{k|k}\big((I-L_{k}C_{k})A_{k-1}e_{k-1|k-1}+r_{k}\big)+\|s_{k}\|_{\Pi_{k|k}}^{2}\\
&=\|F_{k}e_{k-1|k-1}\|^{2}_{\Pi_{k|k}}+r_{k}^{\mathrm{T}}\Pi_{k|k}\big(2F_{k}e_{k-1|k-1}+r_{k}\big)+2s_{k}^{\mathrm{T}}\Pi_{k|k}\big(F_{k}e_{k-1|k-1}+r_{k}\big)+\|s_{k}\|_{\Pi_{k|k}}^{2}\nonumber
\end{align}

First, we establish an upper bound on the first term on the right-hand-side of \eqref{eq:49}.
Considering Lemma \ref{Matrix Inversion Lemma}, \eqref{upd:props4:1}, and the definition of $T_{ii,k}$ in Assumption \ref{weakly}, we can further obtain
\begin{align}\label{upd:thm:2}
  \Pi_{i,k|k}&\leq \Big(F_{ii,k}P_{i,k-1|k-1}F_{ii,k}^{\mathrm{T}}+H_{ii,k}Q_{i}H_{ii,k}^{\mathrm{T}}\Big)^{-1}\nonumber\\
  & = (F_{ii,k}^{\mathrm{T}})^{-1}(\Pi_{i,k-1|k-1}-T_{ii,k})F_{ii,k}^{-1}
\end{align}
By left multiplying $F_{ii,k}^{\mathrm{T}}$ and right multiplying $F_{ii,k}$ on both sides of \eqref{upd:thm:2}, it further derived that
\begin{equation}\label{eq:thm:2_1}
F_{ii,k}^{\mathrm{T}}\Pi_{i,k|k}F_{ii,k} \leq\Pi_{i,k-1|k-1}-T_{ii,k}
\end{equation}
On the other hand, based on \eqref{upd:props4:1} and \eqref{upd:props4:1_1}, we have
\begin{align}\label{eq:thrm_update0}
  F_{ii,k}P_{i,k-1|k-1}F_{ii,k}^{\mathrm{T}}+L_{i,k}RL_{i,k}^{\mathrm{T}} & \geq\big(1+\frac{\underline{l}^{2}\underline{r}}{\bar{p}\bar{f}^{2}}\big)F_{ii,k}P_{i,k-1|k-1}F_{ii,k}^{\mathrm{T}}
\end{align}
Taking into account \eqref{upd:thm:2} and \eqref{eq:thrm_update0}, $T_{ii,k}$ can be bounded as
\begin{equation}\label{eq:thrm1}
 T_{ii,k}>\alpha\Pi_{i,k-1|k-1}
\end{equation}
where $\alpha=\frac{(\underline{l}^{2}\underline{r})/(\bar{p}\bar{f}^{2})}{1+(\underline{l}^{2}\underline{r})/(\bar{p}\bar{f}^{2})}$ as defined in Proposition \ref{upd:props:2}.
By stacking \eqref{eq:thm:2_1} and\eqref{eq:thrm1} for all the subsystems $i$, $i\in\mathbb{N}$, it is further obtained:
\begin{subequations}
  \begin{align}
F_{d,k}^{\mathrm{T}}\Pi_{k|k}F_{d,k}&\leq\Pi_{k-1|k-1}-T_{k}\label{eq:thrm5}\\
  T_{k}&>\alpha\Pi_{k-1|k-1}\label{eq:thrm4}
\end{align}
\end{subequations}
If Assumption \ref{weakly} holds, from \eqref{eq:thrm5}
\begin{align}\label{eq:thrm2}
   \|F_{k}e_{k-1|k-1}\|^{2}_{\Pi_{k|k}}&=\|e_{k-1|k-1}\|^{2}_{F_{d,k}^{\mathrm{T}}\Pi_{k|k}F_{d,k}}+\|e_{k-1|k-1}\|^{2}_{\Lambda_{k}}\nonumber\\
   &\leq\|e_{k-1|k-1}\|^{2}_{\Pi_{k-1|k-1}-\frac{1}{2}T_{k}}
\end{align}
where $\Lambda_{k}=F_{d,k}^{\mathrm{T}}\Pi_{k|k}F_{o,k}+F_{o,k}^{\mathrm{T}}\Pi_{k|k}F_{d,k}+F_{o,k}^{\mathrm{T}}\Pi_{k|k}F_{o,k}$.
Based on \eqref{eq:thrm4}, the first term on the right-hand-side of \eqref{eq:49} is bounded as:
\begin{equation}\label{eq:thrm3}
  \|F_{k}e_{k-1|k-1}\|^{2}_{\Pi_{k|k}}\leq(1-\frac{1}{2}\alpha)\|e_{k-1|k-1}\|^{2}_{\Pi_{k-1|k-1}}
\end{equation}

Next, the second term to the fourth term on the right-hand-side of \eqref{eq:49} can be approximated according to Lemma \ref{lemma1} and \ref{lemma2}. Consequently, it can be obtained that
\begin{equation}\label{eq:thrm6}
  V_{k}(e_{k|k}) \leq(1-\frac{1}{2}\alpha)\|e_{k-1|k-1}\|^{2}_{\Pi_{k-1|k-1}}+\kappa_{1}\|e_{k-1|k-1}\|^{3}+\bar{w}\kappa_{2}+\big((1+\sqrt{n}\bar{l}\bar{c})\bar{w}+\bar{l}\bar{v}\big)\kappa_{3}
\end{equation}
for $\|e_{k-1|k-1}\|\leq\varepsilon$.
Define $\varepsilon^{'}=\min(\varepsilon, \frac{\alpha}{4\bar{p}\kappa_{1}})$. For $\|e_{k-1|k-1}\|\leq\varepsilon^{'}$, one have
\begin{equation}\label{eq:thrm7}
  \kappa_{1}\|e_{k-1|k-1}\|^{3}\leq\kappa_{1}\varepsilon^{'}\|e_{k-1|k-1}\|^{2}\leq\frac{\alpha}{4\bar{p}}\|e_{k-1|k-1}\|^{2}\leq\frac{\alpha}{4}V_{k-1}(e_{k-1|k-1})
\end{equation}
Therefore, from \eqref{eq:thrm6} and \eqref{eq:thrm7}, it is further derived that
\begin{equation}
  V_{k}(e_{k|k})-V_{k-1}(e_{k-1|k-1})\leq -\frac{\alpha}{4}V_{k-1}(e_{k-1|k-1})+\kappa(\bar{w}, \bar{v})
\end{equation}
where $\kappa(\bar{w}, \bar{v})=\bar{w}\kappa_{2}+\big((1+\sqrt{n}\bar{l}\bar{c}\big)\bar{w}+\bar{l}\bar{v})\kappa_{3}$.
Therefore, the condition \eqref{iss} is fulfilled, and $V_{k}$ is an ISS-Lyapunov function for the estimation error \eqref{eq:rs}. By applying Lemma \ref{lemma:iss}, the estimation error dynamics provided by the proposed DEKF algorithm \eqref{edkf} is ISS with respected to the disturbances $w_{k}$ and noise $v_{k}$. $\square$
\end{proof}
\begin{rmk}
{\color{black}Assumptions \ref{assume:continuous} to \ref{assume:2} have been commonly used to analyze the boundedness of the estimation error provided by  extended Kalman filters \cite{reif1999stochastic,kluge2010stochastic, battistelli2016stability}. Assumption \ref{weakly} characterizes the strength of coupling between interacting subsystems within a certain neighborhood containing the linearization point. A similar assumption was made in Zheng et al \cite{zheng2018coupling} and Dunbar \cite{dunbar2007distributed} to investigate distributed control of coupled systems. To facilitate the fulfillment of this assumption, we can decompose the system into smaller subsystems in a way that the interactions among the subsystems are made sufficiently small.}
\end{rmk}
\section*{Applications}
In this section, we apply the proposed DKF and DEKF algorithms for three different cases to demonstrate their effectiveness. Firstly, a numerical linear example is utilized to evaluate the performance of the proposed DKF algorithm. Then, we introduce a chemical process with four continuous stirred tank reactors to verify the efficacy of the proposed DEKF approach. This application is also used to verify the satisfaction of the Assumption \ref{weakly}. Finally, the proposed DEKF algorithm is applied to a large-scale water treatment plant to demonstrate its applicability and superiority.
\subsection*{Case study I: Linear numerical example}
First, let us consider an open-loop unstable linear system in the form of \eqref{cmodel} with the following system matrices:
\begin{equation*}
  A = \left[\begin{array}{cccc}
              0.68 & 0.25 & 0.17 & 0.11 \\
              -0.09 & 0.98 & 0 & -0.13 \\
              0.15 & 0 & 0.9 & -0.6 \\
              0.12 & -0.01 & 0.1 & 0.89
            \end{array}\right], ~C=\left[\begin{array}{cccc}
                                          1 & 0 & 0 & 0 \\
                                          0 & 0 & 1 & 0
                                        \end{array}\right].
\end{equation*}

This system is partitioned into two subsystems, each of which is assigned two states: the first two state variables belong to the first subsystem, and the third and fourth state variables are assigned to the second subsystem. Additive system disturbances and measurement noise with zero mean and standard deviation of 1 were added to the state equation and output measurement equation, respectively. The weighting matrices $Q_{i}$ and $R_{i}$ for each subsystem $i$, $i=1,2$, are selected as the covariance matrices of the process disturbances and measurement noise, respectively.
$P_{i,0|-1}$ in each subsystem $i$ is a diagonal matrix with the main diagonal elements being 100. For an illustrating purpose, the initial state is chosen as $x_{0}=[-7.0047~~9.0089~~6.0012~~-3.0066]$, and the initial guess
is determined as $\hat{x}_{0|-1}=[-7.7052~~9.9089~~6.6013~~-3.3073]$. The proposed DKF method for a linear system, which is presented in \eqref{upd:dkf} is applied to establish two local estimators for the two subsystems. The actual states and the state estimates provided by the proposed DKF are presented in Figure \ref{numerical}. The proposed method can provide accurate estimates of the actual states.

To further compare the performance of the two partition-based distributed estimation methods, we examine three designs: 1) the proposed DKF; 2) the DMHE method in Li et al \cite{li2023iterative} with the same hyper-parameters as adopted in DKF; 3) DMHE in Li et al \cite{li2023iterative} with hyper-parameters that are fine-tuned for this specific numerical example based on trial and error. Accordingly, we conduct 500 simulation runs for each of the three designs, and compare the estimation accuracy by assessing the trajectories of the root mean squared errors (RMSEs) for the three designs. At each sampling instant $k$, the RMSE for the state estimate is computed as:
  \begin{equation*}
  \mathrm{RMSE}(k) = \sqrt{\frac{\|\hat{x}_{k}-x_{k}\|^{2}}{n}}
\end{equation*}

The trajectories of the RMSEs for the three designs are presented in Figure \ref{error_numerical}. The shadowed region illustrates the possible range of the estimation error during 500 simulations. The results indicate that the DMHE algorithm proposed in Li et al \cite{li2023iterative} using the same hyper-parameters as DKF exhibits a relatively higher estimation error compared to the proposed DKF method. Meanwhile, the proposed DKF algorithm not only demonstrates a smaller estimation error but also provides a tighter error bound in comparison to the DMHE counterpart.

\begin{rmk}
{\color{black}The DMHE method proposed in Li et al \cite{li2023iterative} adopts constant weighting matrices for the arrival costs of the local state estimators. To achieve good estimation performance and to ensure the satisfaction of the conditions on the stability of the error dynamics, these parameters need to be tuned based on the user's experience and/or extensive trial and error analysis. In contrast, the proposed DKF updates the estimates of the error covariance matrices in the local estimators at every new sampling instant, enabling it to adapt to less accurate initial guesses of the parameters, which is favorable from an implementation standpoint.}
\end{rmk}

\subsection*{Case study II: a four-reactor chemical process}
In this subsection, we evaluate the performance of the proposed DEKF algorithm and demonstrate the feasibility of Assumption \ref{weakly} by considering a chemical process consisting of four continuous stirred tank reactors (CSTRs). The four reactors are interconnected via material and energy flows. The process involves three parallel irreversible exothermic reactions that convert reactant A into product B, and into undesired side products C and D. Figure \ref{cstr} presents a schematic of the process for reference. We apply the proposed DEKF and develop four local estimators for the four subsystems within a distributed framework.

The chemical process under consideration involves a feed stream containing pure reactant A, which enters the first reactor at flow rate $F_{01}$, temperature $T_{01}$, and molar concentration $C_{A01}$. The outlet stream from the first reactor, along with additional pure reactant A, is then fed into the second reactor, with flow rates $F_{1}$ and $F_{02}$, temperatures of $T_{1}$, and $T_{02}$, and molar concentrations of $C_{A1}$ and $C_{A02}$, respectively. A portion of the effluent from the second reactor is recycled back to the first reactor at a flow rate of $F_{r1}$, temperature of $T_{2}$, and concentration of $C_{A2}$. The remaining effluent enters the third reactor at a flow rate of $F_{2}$, temperature of $T_{2}$, and concentration of $C_{A2}$. Feed streams that contain pure A are also introduced in the third and fourth reactors, with flow rates $F_{03}$ and $F_{04}$, temperatures of $T_{03}$ and $T_{04}$, and molar concentrations of $C_{A03}$ and $C_{A04}$, respectively. Additionally, a portion of the effluent from the fourth reactor is recycled back to the first reactor at a flow rate of $F_{r2}$, temperature of $T_{4}$, and concentration of $C_{4}$. Each reactor is equipped with a jacket, which can be used to provide or remove heat to or from the corresponding reactor. A more comprehensive description of the process, the definitions of the parameters, and a dynamical process model can be found in Rashedi et al \cite{rashedi2017triggered} and Yin and Liu \cite{yin2020distributed}.

In the simulation, we consider a steady-state operating point $x_{s}=[310.84~\mathrm{K}~~3.03~\mathrm{kmol/m^{3}}$ $310.83~\mathrm{K}~~2.80~\mathrm{kmol/m^{3}}~~312.47~\mathrm{K}~~2.84~\mathrm{kmol/m^{3}}~~311.16~\mathrm{K}~~3.01~\mathrm{kmol/m^{3}}]$, which is corresponding to open-loop time-varying heating inputs applied to the four reactors $Q_1 = (1+2\sin(0.2\pi t))\times 10^4~\mathrm{kJ/h}$, $Q_2=(2+4\sin(0.2\pi t)\times 10^4~\mathrm{kJ/h}$, $Q_3=(2.5+5\sin(0.2\pi t))\times 10^4~\mathrm{kJ/h}$, and $Q_4=(1+2\sin(0.2\pi t))\times 10^4~\mathrm{kJ/h}$. The initial state $x_{0}$ is shown in Table \ref{tbl:initial_states}, where $C_{Ai}$ and $T_{i}$ represent the molar concentration of A and the temperature in the $i$th reactor, $i=1,2,3,4$.
Unknown additive system disturbances and measurement noise that follow Gaussian distribution with zero mean and standard deviation of $0.001\times x_{s}$ and $0.001\times Cx_{s}$ are introduced to the states and measured outputs, respectively.

\subsubsection*{Results on the four-CSTR chemical process}
The process is partitioned into four subsystems according to the physical topology: each subsystem accounts for one reactor. The proposed DEKF is applied to the process, and four estimators are developed based on the four subsystems.
The weighting matrix $Q_{i}$ of each estimator for subsystem $i$, $i=1,2,3,4$, is chosen as a diagonal matrix with the main diagonal elements being 150; the weighting matrix $R$ is selected as a diagonal matrix with the main diagonal elements being 1; $P_{i,0|-1}$, $i=1,2,3,4$, is a diagonal matrix with the main diagonal elements being 0.01. The initial guess of the eight process states that will be adopted by the proposed DEKF is given in Table \ref{tbl:initial_states}. The trajectories of the actual state and the state estimates provided by the proposed DEKF method are presented in Figure \ref{cstr_results}. The results confirm that the proposed DEKF algorithm can give accurate estimates of the actual states of this nonlinear chemical process.

Through the simulations, we record all {\emph{posteriori}} estimates of the covariance matrices, i.e., $P_{i,k|k}$ for $k\geq0$, and record the matrices of the linearized system and subsystems, i.e., $C_{[:,i],k}$ and $A_{li,k}$. With these matrices, we verify that Assumption \ref{weakly} is satisfied at every sampling instant during the entire operating period, as the developed DEKF is implemented for the nonlinear chemical process. This is important for confirming the validity of the assumption and the associated theoretical findings on the stability of the proposed method for general nonlinear processes.

\subsection*{Case study III: application to wastewater treatment plant}
In this subsection, we apply the proposed DEKF method for the state estimation of a wastewater treatment plant (WWTP) based on the Benchmark Simulation Model No.1 (BSM1) \cite{alex2008benchmark}. A schematic diagram of which is presented in Figure \ref{WWTP}.
The plant comprises a five-chamber biological reactor and a secondary settler. The biological reactor contains two sections: a non-aerated section constituted by the first two anoxic chambers, and an aerated section consisting of the remaining three aerobic chambers. The non-aerated section is designed to facilitate pre-denitrification reactions in which nitrate is converted into nitrogen. At the same time, the nitrification process, which involves the oxidation of ammonium into nitrite, takes place in the aerated section.


In this wastewater treatment process, the inlet wastewater along with two recycle streams, enters the first chamber of the biological reactor at concentration $Z_{0}$ and flow rate $Q_{0}$. A portion of the outlet stream from the fifth aerobic chamber, with a flow rate of $Q_{f}$ and concentration of $Z_{f}$, is directed into the settler while the remaining portion of this outlet stream with concentration $Z_{a}$ is recycled back to the first chamber at flow rate $Q_{a}$. The secondary settler comprises 10 nonreactive layers, among which the fifth layer act as the feed layer. The settler effluent is composed of three portions: 1) the overflow of the settler containing purified water which is removed consecutively through the first layer at flow rate $Q_{e}$ and concentration $Z_{e}$; 2) a portion of the underflow of the settler, which is recycled back to the first chamber of the biological reactor at concentration $Z_{r}$ and flow rate $Q_{r}$; 3) the rest of the underflow, which is discharged through the bottom layer of the settler at flow rate $Q_{w}$ and concentration $Z_{w}$.

In the biological reactor of this plant, eight biological reaction processes and 13 major compounds are considered. The concentrations of these compounds in the five chambers constitute the state variables of the five-chamber biological reactor. The definitions of state variables in each chamber of the reactor are given in Table \ref{tb1}. In terms of the settler, the dynamical behaviors in each of the ten layers are described by 8 state variables, including $S_{O}$, $S_{ALK}$, $S_{NH}$, $S_{NO}$, $S_{S}$, $S_{I}$, $S_{ND}$ and $X$. The concentration of suspended solids in each layer, denoted by $X$, is calculated as the sum of $X_{s}$, $X_{I}$, $X_{B_{A}}$, $X_{B_{H}}$, $X_{P}$, and $X_{ND}$ in the corresponding layer. Therefore, in total 145 states variables are used to describe the dynamic behaviors of the entire wastewater treatment process. Detailed descriptions of this wastewater treatment process, a dynamic process model, and the definitions of the parameters involved are referred to \cite{alex2008benchmark}.

In each chamber of the reactor, eight physical sensors are deployed to collect output measurements, including the concentrations of dissolved oxygen, free and saline ammonia (i.e., NH3 and $\text {NH}4^{+}$), nitrate and nitrate nitrogen, alkalinity, chemical oxygen demand (COD), filtered chemical oxygen demand ($\mathrm{COD}_{f}$), biological oxygen demand (BOD) and the concentration of suspended solids \cite{yin2018subsystem, yin2019subsystem}. A more detailed explanation of the output measurements for the biological reactor is presented in Table~\ref{tb2}. In the secondary settler, the states in the top layer and the bottom layer are measured.


\subsubsection*{Simulation setting}
We consider the operation of the wastewater treatment process under three different weather conditions.
The data profiles of the wastewater fed into the WWTP under different weather conditions were obtained from the International Water Association website \cite{wastewater}. The data profiles include information on the inlet flow rate $Q_{0}$, and the corresponding concentration $Z_{0}$ under dry, rainy, and stormy weather conditions. The key model parameters for the process are set to be identical to those reported in Alex et al \cite{alex2008benchmark}. The process is decomposed into three subsystems for distributed state estimation, following the decomposition recommended by Yin et al \cite{yin2018subsystem}.
The sampling period for output measurement collection is selected as 15 min. Let $x^{i}_{s}$ denote the steady-state of the $i$th subsystem, $i = 1,2,3$.
In the simulation of the process model, the initial condition for each subsystem $i$ is chosen as $x_{0}^{i}=1.02\times x^{i}_{s}$.
Unknown additive system disturbances following Gaussian distribution with zero mean and a standard deviation of $0.001\times x^{i}_{0}$ are made bounded within $[-0.005\times x^{i}_{0},~0.005\times x^{i}_{0}]$ and
added to the state equations for the $i$th subsystem; random measurement noise for subsystem $i$, $i=1,2,3$, is considered to be Gaussian white noise with zero mean and standard deviation of $0.001\times C_{ii}x^{i}_{0}$ and is bounded within $[-0.005\times C_{ii}x^{i}_{0},~0.005\times C_{ii}x^{i}_{0}]$.
The weighting matrices $Q_{i}$ and $R_{i}$ of estimator $i$ are chosen as the diagonal matrices with the main diagonal elements being 0.5; $P_{i,0|-1}$ of the $i$th estimator, $i=1,2,3$, is a diagonal matrix with main diagonal elements being 0.01. With the selected parameters, a DEKF scheme consisting of three local estimators based on the three decomposed subsystem models is developed for the wastewater treatment process.


\subsubsection*{Results under dry weather condition}
The trajectories of the actual states and corresponding state estimates, for selected state variables, provided by the proposed DEKF for the dry weather condition are presented in Figure \ref{dry}. The results confirm that the proposed DEKF method can provide accurate estimates of the process states under dry weather.

In addition, we compared the performance of the proposed DEKF algorithm with a distributed moving horizon estimation (DMHE) approach presented in Yin et al \cite{yin2018subsystem} under dry weather condition.
In the proposed distributed estimation design, the local estimators receive and leverage output measurement information from the interacting subsystems when computing to provide state estimates of the corresponding local subsystem, while in DMHE design \cite{yin2018subsystem}, only the output measurements of the corresponding subsystem are used by each local estimator. In terms of the DMHE scheme based on Yin et al \cite{yin2018subsystem}, the length of the estimation horizon for each of the local estimators is set as $N_{1}=N_{2}=10$ and $N_{3}=20$; the weighting matrices $Q_{i}$ and $R_{i}$ are selected as the diagonal matrices with the main diagonal elements being 0.01 in a scaled setting; the weighting matrix $P_{i}$ is chosen as a diagonal matrix with the main diagonal elements being 0.5. The estimators of the DMHE use the same initial guesses $\hat{x}_{0|-1}$, $i=1,2,3$, as adopted by the proposed DEKF algorithm. The trajectories of the estimates given by the DMHE approach in Yin et al \cite{yin2018subsystem} are also shown in Figure \ref{dry}. The trajectories of the RMSE for estimates for the proposed DEKF approach and the DMHE approach in Yin et al \cite{yin2018subsystem} are presented in Figure \ref{fig:error}.
The proposed DEKF approach outperforms the DMHE scheme in terms of estimation accuracy. This is attributed to the fact that the DMHE approach in Yin et al \cite{yin2018subsystem} relies only on the output measurements of the local subsystem to generate subsystem state estimates, yet it disregards all the information contained in sensor measurements from the interacting subsystems. It is also worth mentioning that the computation time needed for executing the proposed DEKF has been significantly reduced as compared to the DMHE, since no optimization is involved in this recursive approach.

\subsubsection*{Results under rainy and storm weather conditions}

The proposed DEKF approach is further applied for estimating the process states under rainy and stormy weather conditions.
The parameters of the estimators used for the dry weather condition are adopted for the rainy and stormy weather. Additional simulations are conducted, and the corresponding state estimation results are obtained. The trajectories of the actual states and the obtained state estimates for selected variables under the two weather conditions are presented in
Figures \ref{rain} and \ref{storm}, respectively. The proposed DEKF method can provide accurate estimates that can track the actual states in the presence of unknown disturbances.

\section*{Data Availability and Reproducibility Statement}
The numerical data for Figures \ref{numerical}, \ref{error_numerical}, \ref{cstr_results}, and \ref{dry}-\ref{storm} are provided in the Supplementary Material. The compressed file contains the steady-state values for the examples, the estimation results (including simulated process states and the estimates of these states), and the data profiles for the inlet flow rate $Q_{0}$, and the corresponding concentration $Z_{0}$ under dry, rainy, and stormy weather conditions that are related to the application to the wastewater treatment process. The data can be used to generate the figures presented in this paper.
Particularly, Figure \ref{error_numerical} presents the RMSEs for the proposed DKF algorithm, DMHE in Li et al \cite{li2023iterative}, and DMHE in Li et al \cite{li2023iterative} with fine-tuned hyper-parameters. We provide data on the estimation results for the three methods obtained from 500 simulation runs. Figure 3 can be generated based on the RMSEs for the 500 simulation runs. The data for Figure \ref{dry}, Figure \ref{rain}, and Figure \ref{storm} can be obtained by using the steady-state and input data profiles provided in the Supplementary Material, and following the simulation settings described in this paper.
Figure \ref{fig:error} shows the trajectories of the RMSE of estimation error for the proposed DEKF algorithm and the distributed moving horizon estimation method in Yin et al\cite{yin2018subsystem}.

\section*{Conclusions}
We proposed a partition-based distributed extended Kalman filter (DEKF) algorithm that enables efficient and scalable distributed state estimation for general nonlinear systems involving interconnected subsystems. A distributed Kalman filter design based on the subsystem models of a general linear system was established by exploring the equivalence between Kalman filtering and full-state information within the distributed context. Then, the linear design was extended to nonlinear systems, and a DEKF method was formulated. The stability of the proposed method was proven under mild and verifiable assumptions in a nonlinear setting. Extensive simulations were conducted to demonstrate the efficacy and superiority of the proposed method. Specifically, an open-loop unstable numerical example was used to verify the effectiveness of the proposed method in the linear setting and showcase its advantage over a distributed moving horizon estimation design. The effectiveness of the proposed DEKF, and the viability of a key assumption made in this work were confirmed by using a chemical process example consisting of four interconnected reactors. The proposed approach was also applied to a large-scale water treatment plant, and good estimates of all the 145 state variables were obtained. Through a comparative analysis, the proposed DEKF method was demonstrated to outperform a DMHE scheme developed for the state estimation of WWTP.
\section*{Acknowledgment}
This research is supported by Ministry of Education, Singapore, under its Academic Research Fund
Tier 1 (RG63/22), and Nanyang Technological University, Singapore (Start-Up Grant).

\newpage
\renewcommand\refname{Literature Cited}

\newpage~
\begin{figure}[htbp]
  \centering
  \includegraphics[width=0.65\textwidth]{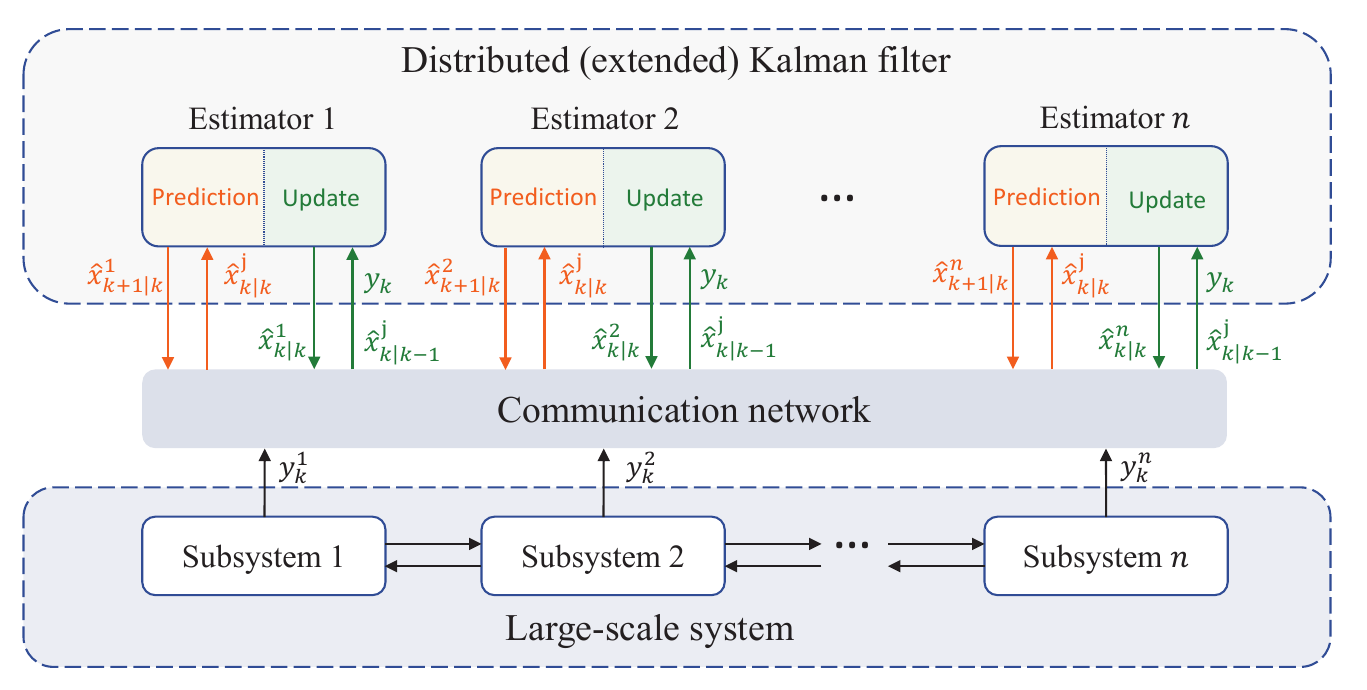}
  \caption{A schematic diagram of proposed partition-based distributed (extended) Kalman filter scheme}\label{fig:dkf}
\end{figure}

\newpage~
\begin{figure}[t]
  \centering
  \includegraphics[width=0.79\textwidth]{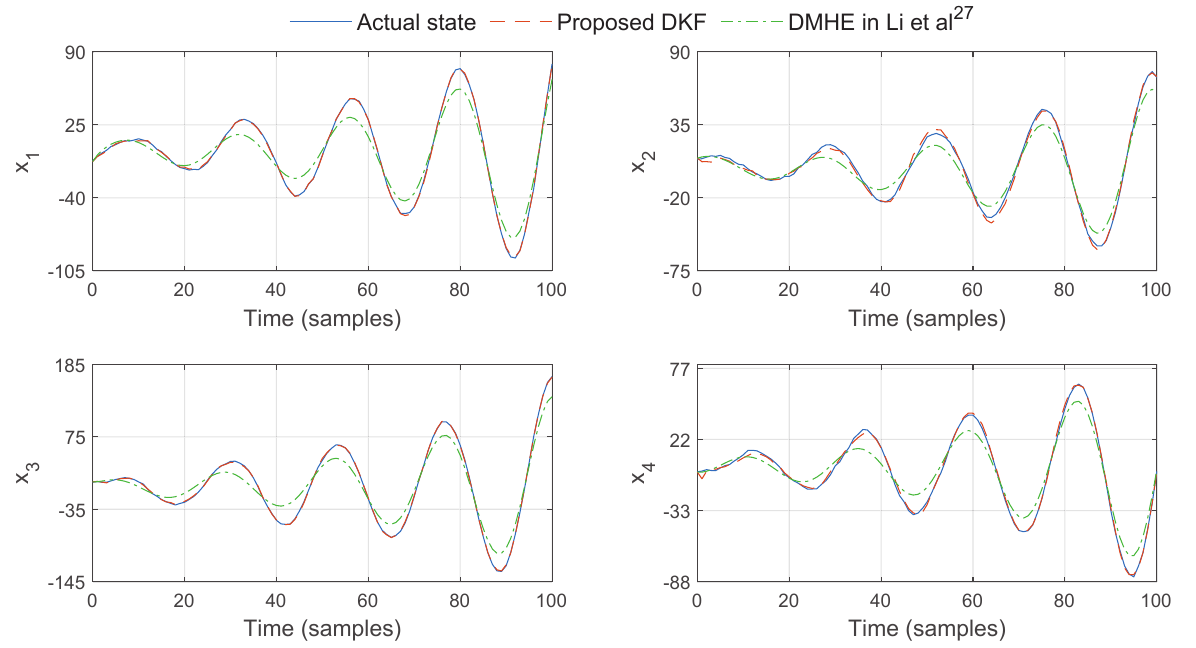}
  \caption{The trajectories of the actual states and the state estimates for the numerical example.}\label{numerical}
\end{figure}

\newpage~
\begin{figure}[htbp]
  \centering
  \includegraphics[width=0.75\textwidth]{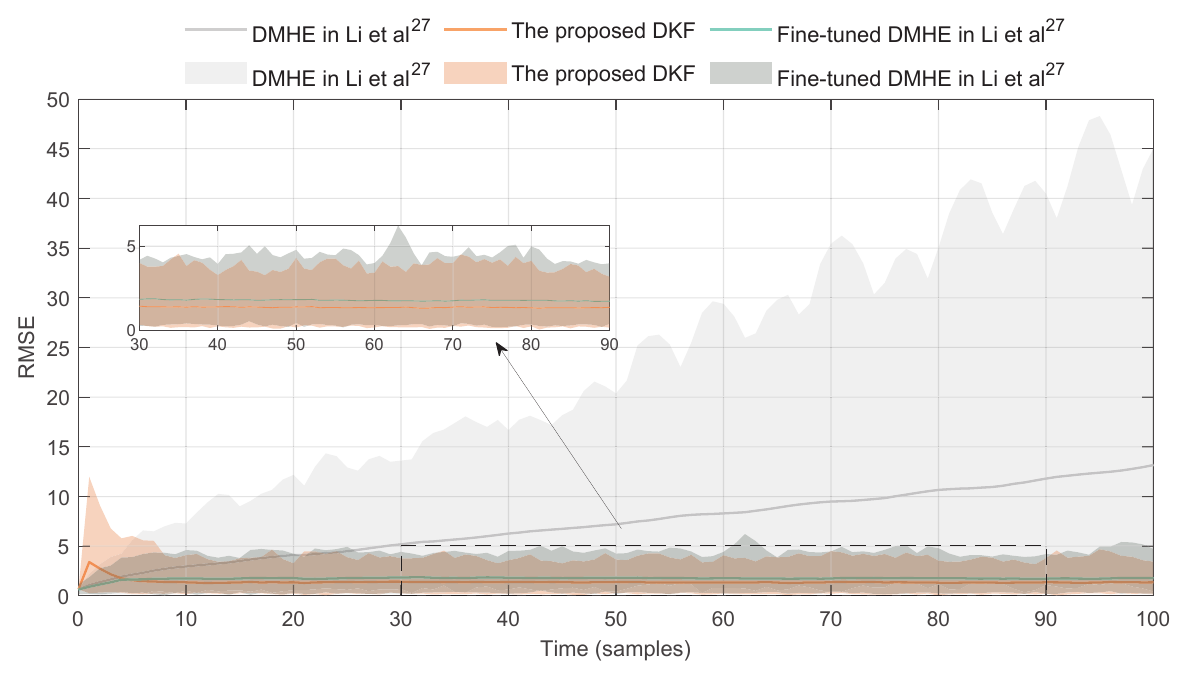}
  \caption{The trajectories of the RMSEs for the three designs: 1) proposed DKF; 2) DMHE in Li et al \cite{li2023iterative} with same hyper-parameters as DKF; 3) DMHE in Li et al \cite{li2023iterative} with fine-tuned hyper-parameters. The shadowed region shows the possible range of the RMSE over 500 simulations.}\label{error_numerical}
\end{figure}

\newpage~
\begin{figure}[tttt]
  \centering
  \includegraphics[width=0.87\textwidth]{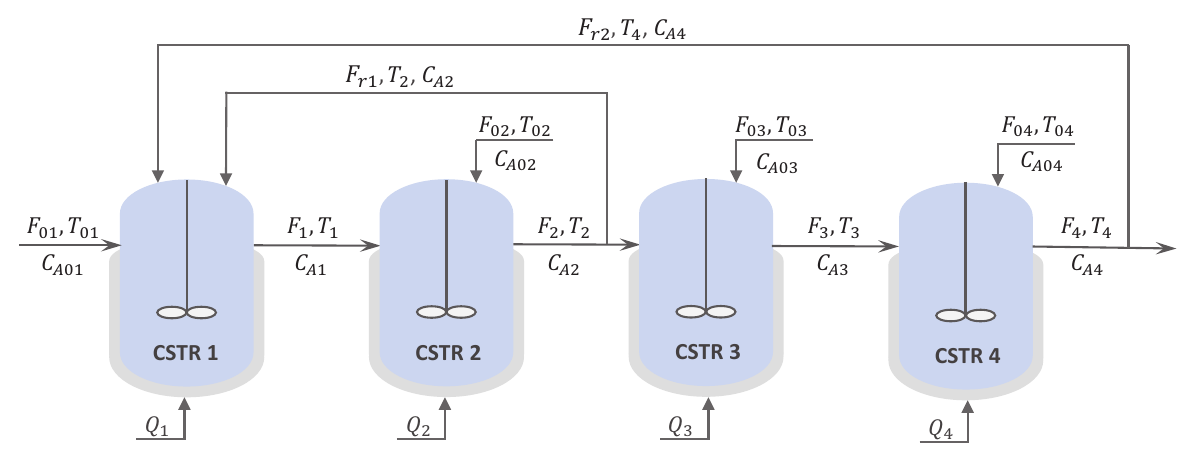}
  \caption{A schematic of the four-CSTR process.}\label{cstr}
\end{figure}

\newpage~
\begin{figure}[t]
  \centering
  \includegraphics[width=0.79\textwidth]{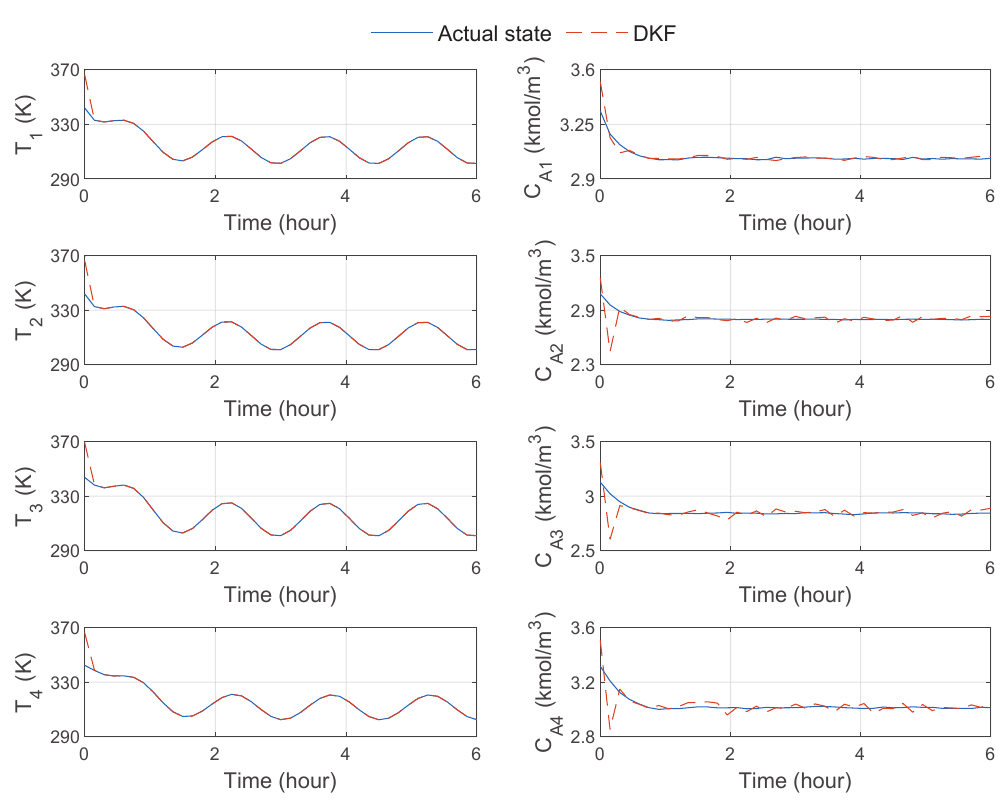}
  \caption{The trajectories of the actual states and the state estimates for four reactors.}\label{cstr_results}
\end{figure}

\newpage~
\begin{figure}[t]
  \centering
  \includegraphics[width=1\textwidth]{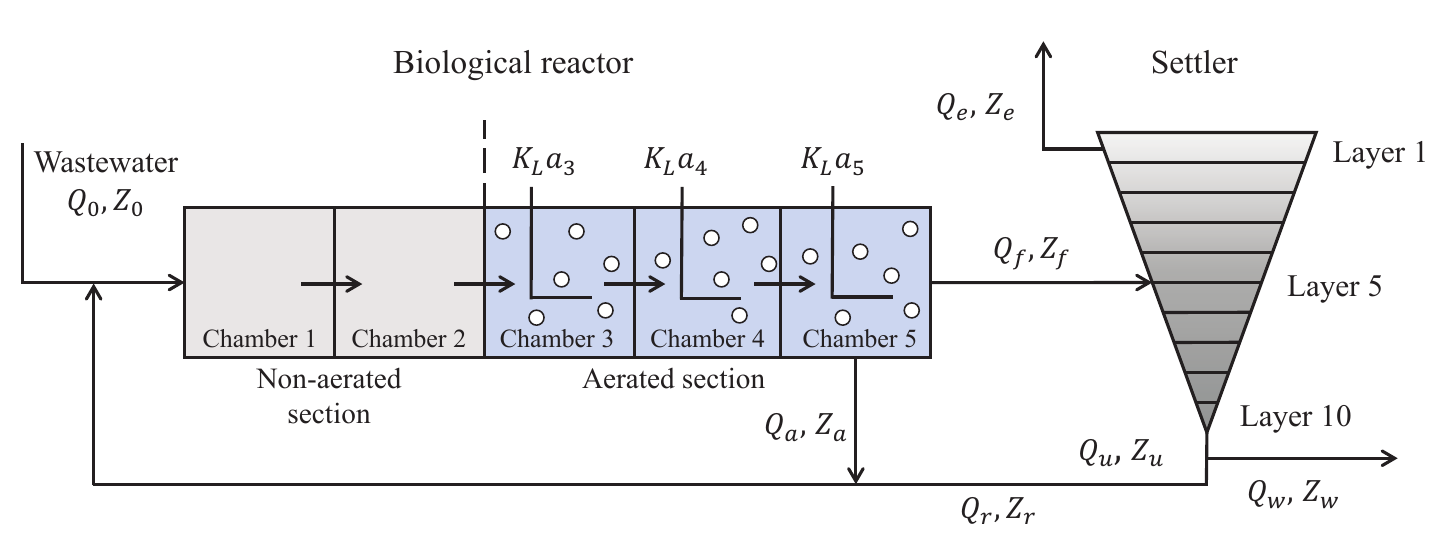}
  \caption{A schematic of the wastewater treatment plant.}\label{WWTP}
\end{figure}

\newpage~
\begin{figure}[t]
  \centering
  \includegraphics[width=0.89\textwidth]{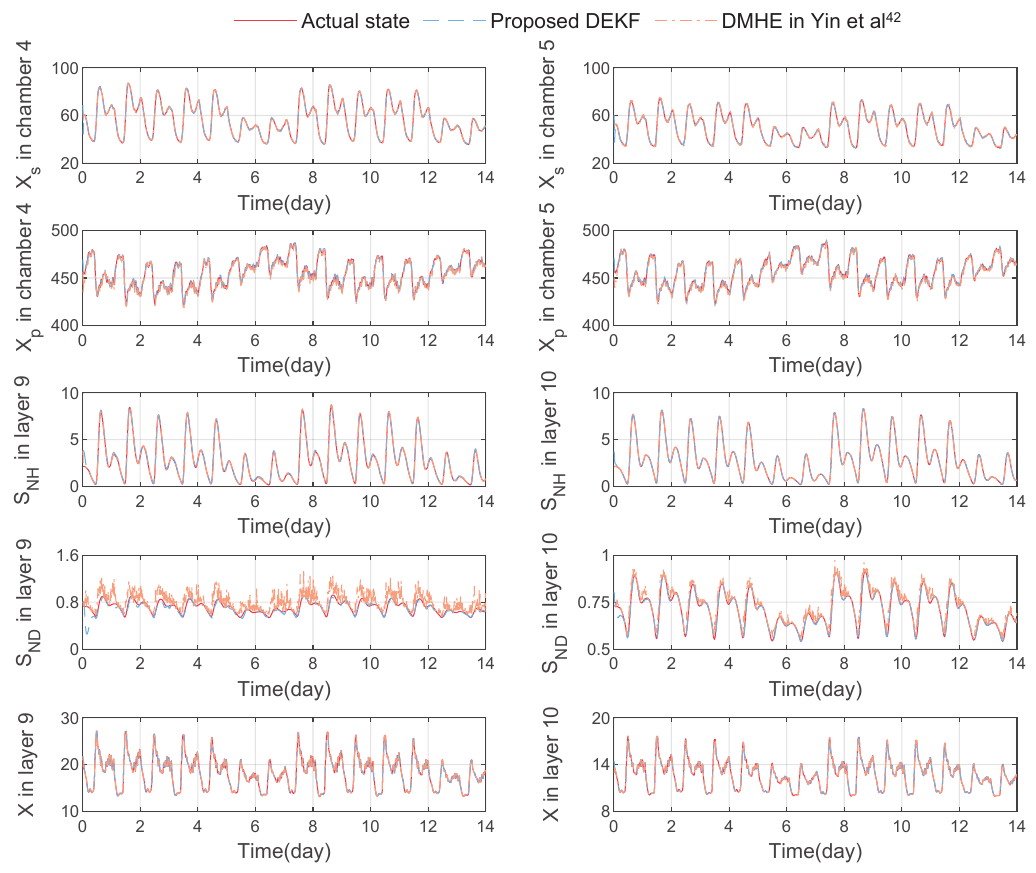}
  \caption{The trajectories of the actual states and the state estimates in dry weather.}\label{dry}
\end{figure}

\newpage~
\begin{figure}[t]
  \centering
  \includegraphics[width=0.71\textwidth]{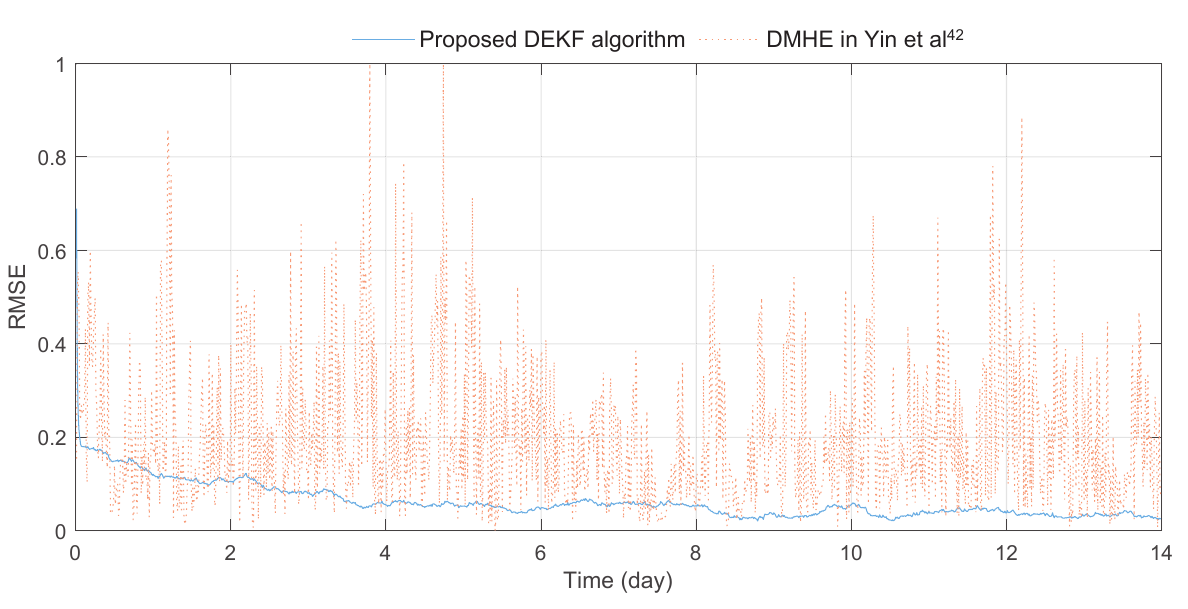}
  \caption{The trajectories of the estimation error given by the proposed DEKF algorithm and the distributed moving horizon estimation in Yin et al \cite{yin2018subsystem}.}\label{fig:error}
\end{figure}

\newpage~
\begin{figure}[t]
  \centering
  \includegraphics[width=0.92\textwidth]{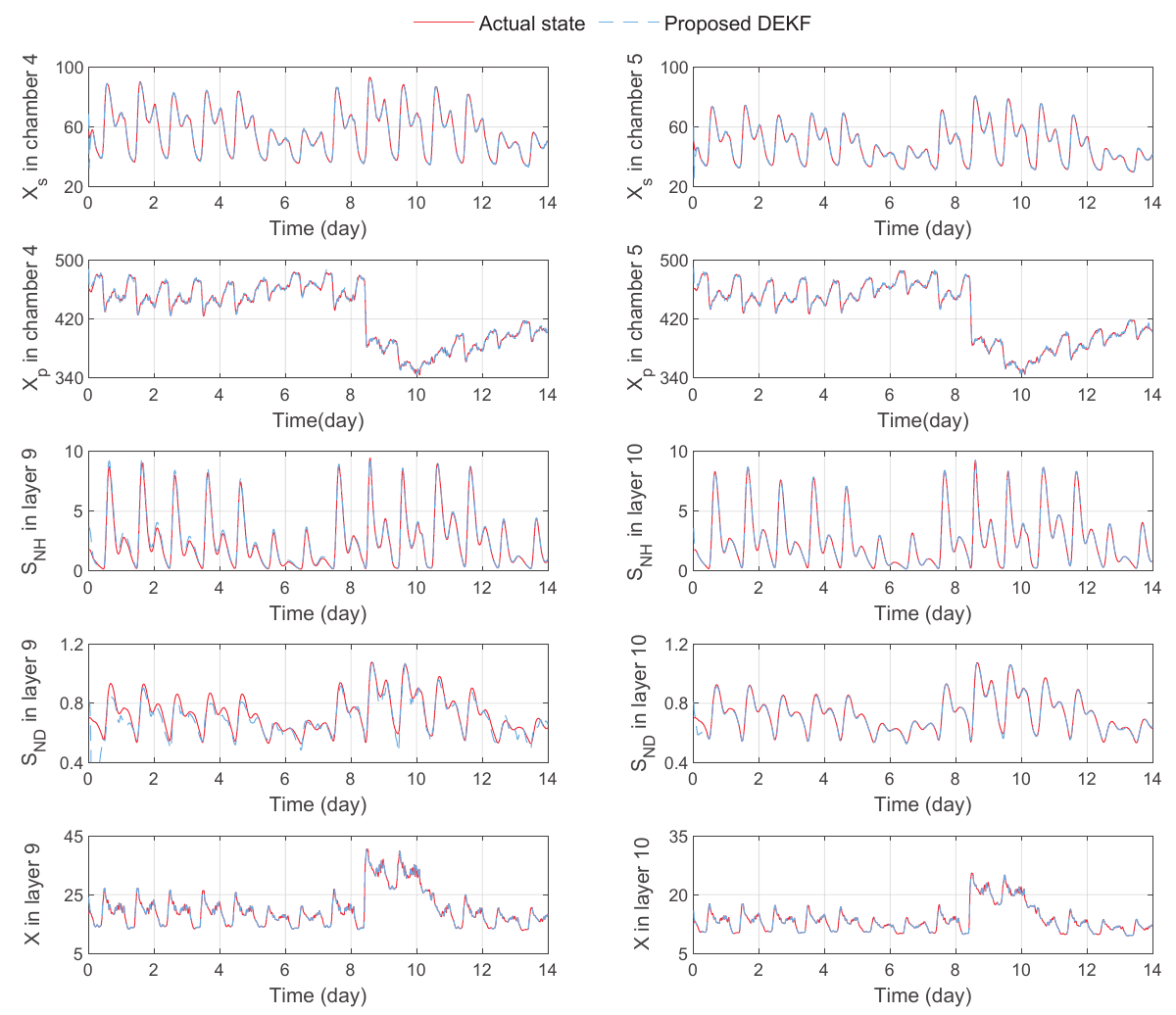}
  \caption{The trajectories of the actual states and the state estimates under rainy weather.}\label{rain}
\end{figure}

\newpage~
\begin{figure}[t]
  \centering
  \includegraphics[width=0.89\textwidth]{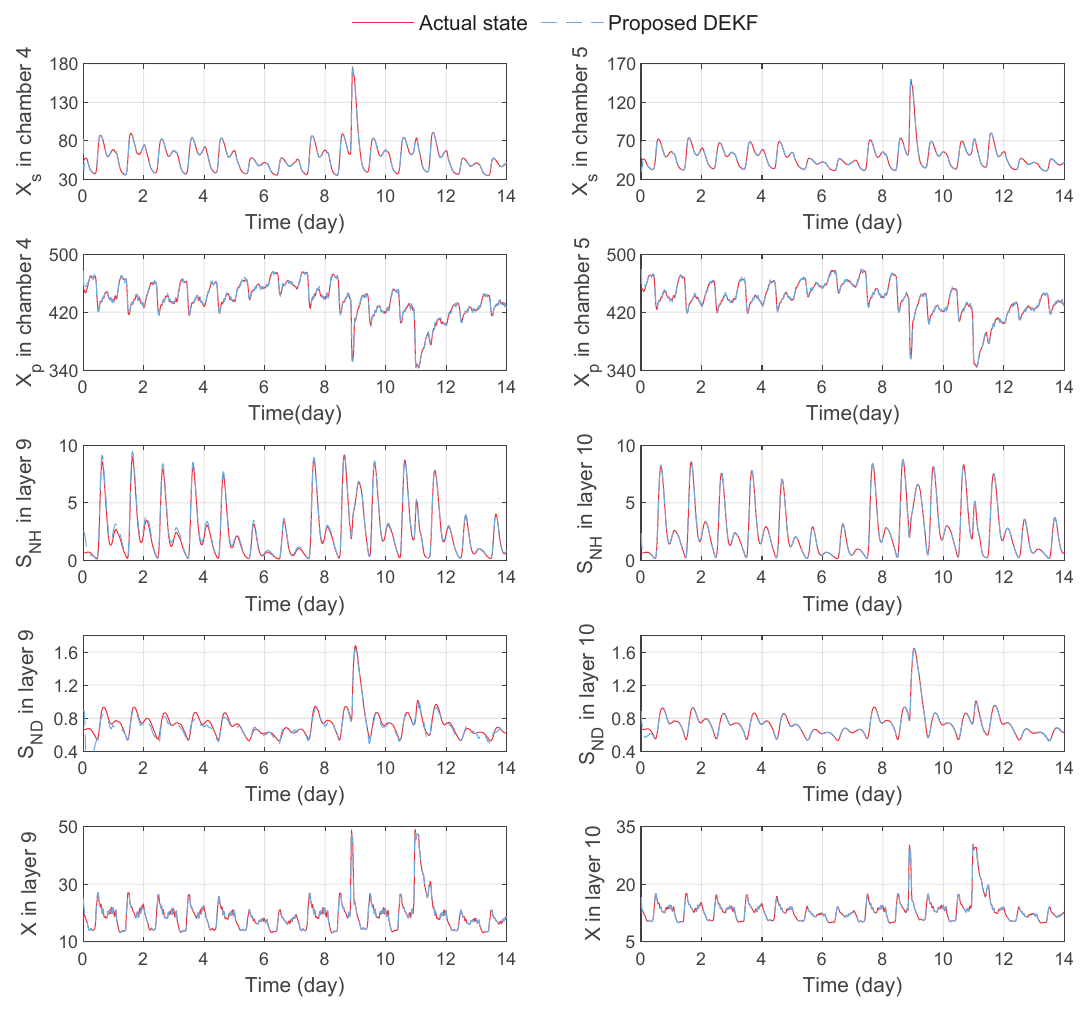}
  \caption{The trajectories of the actual states and the state estimates under stormy weather.}\label{storm}
\end{figure}

\newpage~
\begin{table}
  \centering
    \caption{The actual value $x_{0}$ and the initial guess $\hat{x}_{0}$ of the initial stage of the process}\label{tbl:initial_states}
  \begin{tabular}{cccccccccc}
   \toprule[1pt]
   & $T_{1}$   & $C_{A1}$ & $T_{1}$  & $C_{A2}$ & $T_{3}$  & $C_{A3}$  & $T_{4}$  & $C_{A4}$ \\
   \midrule
   $x_{0}$ & 341.9213 & 3.3349 & 341.9161 & 3.0803 & 343.7129 & 3.1286 & 342.2733 & 3.3156 \\
    $\hat{x}_{0}$ & 362.7388 & 3.5215 & 362.8758 & 3.2521 & 365.0446  & 3.3023 & 362.5779 & 3.5010 \\
   \bottomrule[1pt]
\end{tabular}
\end{table}

\newpage~
\begin{table}
  \centering
    \caption{Definitions and physical meanings of key state variables of the WWTP.}\label{tb1}
  \begin{tabular}{lll}
  \toprule
  State & Physical meaning & Unit\\
  \midrule
  $S_{I}$  &  Inert soluble organic matter & $\mathrm{g ~COD ~m^{-3}}$\\
  $S_{S}$ &   Readily biodegradable and soluble substrate &$\mathrm{g ~COD ~m^{-3}}$\\
  $X_{I}$ &    Inert particulate organic matter  &$\mathrm{g~ COD ~m^{-3}}$\\
    $X_{S}$  &  Slowly biodegradable and soluble substrate &$\mathrm{g ~COD ~m^{-3}}$ \\
  $X_{B_{H}}$ &   Biomass of active heterotrophs  &$\mathrm{g ~COD~ m^{-3}}$\\
  $X_{B_{A}}$ &    Biomass of active autotrophs  &$\mathrm{g ~COD ~m^{-3}}$\\
    $X_{P}$  &  Particulate generated from decay of organisms & $\mathrm{g~ COD ~m^{-3}}$\\
  $S_{O}$ &   Dissolved oxygen  &$\mathrm{g ~(-COD)~ m^{-3}}$\\
  $S_{NO}$ &    Nitrate nitrogen and nitrate  &$\mathrm{g~N~m^{-3}}$\\
    $S_{NH}$  &  Free and saline ammonia & $\mathrm{g~N~m^{-3}}$\\
  $S_{ND}$ &   Biodegradable and soluble organic nitrogen  &$\mathrm{g~N~m^{-3}}$\\
  $X_{ND}$ &    Particulate biodegradable organic nitrogen  & $\mathrm{g~N~m^{-3}}$\\
      $S_{ALK}$ &    Alkalinity  & $\mathrm{mol~ m^{-3}}$\\
  \bottomrule
\end{tabular}
\end{table}

\newpage~
\begin{table}
  \centering
    \caption{Relationship between the state variables output measurements for each chamber of the biological reactor. }\label{tb2}
  \begin{tabular}{lll}
  \toprule
  Measurement & Expression in the form of state variables  \\
  \midrule
  Concentration of dissolved oxygen & $S_{O}$ \\
   Concentration of free and saline ammonia & $S_{NH}$   \\
   Concentration of nitrate and nitrate nitrogen &  $S_{NO}$  \\
     Concentration of alkalinity &  $S_{ALK}$ \\
  COD &  $S_{S}+S_{I}+X_{S}+X_{I}+X_{B_{A}}+X_{B_{H}}$ \\
  $\mathrm{COD}_{f}$ & $S_{S}+S_{I}$  \\
    BOD  & $S_{S}+X_{S}$ \\
  Concentration of suspended solids &  $X_{S}+X_{I}+X_{B_{A}}+X_{B_{H}}+X_{P}+X_{ND}$\\
  \bottomrule
\end{tabular}
\end{table}
\end{document}